\documentclass{article}

\usepackage{microtype}
\usepackage{graphicx}
\usepackage{subfigure}
\usepackage{booktabs} % for professional tables

\usepackage{hyperref}

\usepackage{lipsum}
\usepackage{wrapfig}
\usepackage{amsmath}
\usepackage{amsthm}
\usepackage{amssymb}
\newcommand{\bfx}{\mathbf{x}}

\newcommand{\bfr}{\mathbf{r}}
\newcommand{\bfe}{\mathbf{e}}
\newcommand{\del}[2]{\frac{\partial #1}{\partial #2}}

\DeclareMathOperator{\Dir}{Dir}
\DeclareMathOperator{\Cat}{Cat}

\DeclareMathOperator{\CrossEntropy}{CrossEntropy}

\DeclareMathOperator{\Vol}{Vol}
\usepackage{multirow,makecell}

\newcommand{\bfp}{\mathbf{p}}
\usepackage[accepted]{icml2024}

\usepackage{amsmath}
\usepackage{amssymb}
\usepackage{mathtools}
\usepackage{amsthm}

\usepackage[capitalize,noabbrev]{cleveref}

\theoremstyle{plain}

\newtheorem{proposition}{Proposition}
\newtheorem{prop}{Proposition}

\theoremstyle{definition}

\theoremstyle{remark}

\usepackage[textsize=tiny]{todonotes}

\icmltitlerunning{Dirichlet Flow Matching with Applications to DNA Sequence Design}
\begin{document}

\twocolumn[
\icmltitle{Dirichlet Flow Matching with Applications to DNA Sequence Design}


\icmlsetsymbol{equal}{*}

\begin{icmlauthorlist}
\icmlauthor{Hannes Stark}{equal,yyy}
\icmlauthor{Bowen Jing}{equal,yyy}
\icmlauthor{Chenyu Wang}{yyy}
\icmlauthor{Gabriele Corso}{yyy}\\
\icmlauthor{Bonnie Berger}{yyy,zzz}
\icmlauthor{Regina Barzilay}{yyy}
\icmlauthor{Tommi Jaakkola}{yyy}
\end{icmlauthorlist}

\icmlaffiliation{yyy}{CSAIL, Massachusetts Institute of Technology}
\icmlaffiliation{zzz}{Dept. of Mathematics, Massachusetts Institute of Technology}

\icmlcorrespondingauthor{Hannes Stark}{hstark@mit.edu}
\icmlcorrespondingauthor{Bowen Jing}{bjing@mit.edu}

\icmlkeywords{Machine Learning, ICML}

\vskip 0.3in
]



\printAffiliationsAndNotice{\icmlEqualContribution} % otherwise use the standard text.
%  and evaluated on proteins with multiple recent structures %  the computational understanding of protein structure. 
\begin{abstract}
Discrete diffusion or flow models could enable faster and more controllable sequence generation than autoregressive models. We show that na\"ive linear flow matching on the simplex is insufficient toward this goal since it suffers from discontinuities in the training target and further pathologies. To overcome this, we develop \emph{Dirichlet flow matching} on the simplex based on mixtures of Dirichlet distributions as probability paths. In this framework, we derive a connection between the mixtures' scores and the flow's vector field that allows for classifier and classifier-free guidance. Further, we provide distilled Dirichlet flow matching, which enables one-step sequence generation with minimal performance hits, resulting in $O(L)$ speedups compared to autoregressive models. On complex DNA sequence generation tasks, we demonstrate superior performance compared to all baselines in distributional metrics and in achieving desired design targets for generated sequences.  Finally, we show that our classifier-free guidance approach improves unconditional generation and is effective for generating DNA that satisfies design targets. Code is available at \url{https://github.com/HannesStark/dirichlet-flow-matching}.
\end{abstract}

\section{Introduction}

% However, despite the wealth of works on adapting diffusion models to discrete categorical data 
Flow matching (FM) is a generative modeling framework that provides a simulation-free means of training continuous normalizing flows (CNFs) between noise and data distributions \citep{lipman2022flow, liu2022flow, albergo2022building}. Such flow models can be viewed as generalizing diffusion models \citep{song2021score} to permit more flexible design of iterative noising processes, which have proven useful in generative modeling on non-euclidean spaces such as compact Riemannian manifolds \citep{chen2023riemannian, bose2023se}. However, existing formulations of flow matching have focused on continuous spaces and have yet to treat \emph{discrete categorical data}---a notable shortcoming considering the many important applications of discrete generative modeling such as text generation and biological sequence design \citep{avdeyev2023dirichlet}.

%extant works on discrete diffusion \citep{austin2021structured, hoogeboom2021autoregressive, dieleman2022continuous, campbell2022continuous, avdeyev2023dirichlet}. Leveraging the flexibility and performance of flow matching in these modalities would benefit many applications 

In this work, we introduce \emph{Dirichlet flow matching} for generative modeling of discrete categorical data and sequences over such data. We frame such modeling as a transport problem between a uniform density over the probability simplex and a finitely supported distribution over the vertices of the simplex. We then learn a continuous vector field to parameterize this transport based on a conditional noising process that generates a time-evolving Dirichlet distribution. Na\"ively, it may appear more straightforward to define a noising process that linearly interpolates between data and noise, as is the dominant approach in flow matching on $\mathbb{R}^n$ \citep{liu2022flow, lipman2022flow, pooladian2023multisample} and which can be trivially adapted to the simplex. However, we show that such an approach---which we call linear flow matching---suffers from pathological behavior due to the contracting support of the resulting conditional probability paths. We carefully engineer Dirichlet FM to avoid these shortcomings while retaining the advantages of flow matching by being significantly simpler---both computationally and conceptually---than DDSM \citep{avdeyev2023dirichlet}, a discrete diffusion model that also relaxes data onto the simplex. %Furthermore, Dirichlet FM does not require an ordering of the classes of the categorical variable, unlike the stick-breaking construction of DDSM. 
Thanks to these properties, we posit (and confirm) that Dirichlet FM is significantly more fit for modeling complex discrete distributions than either linear FM or DDSM.

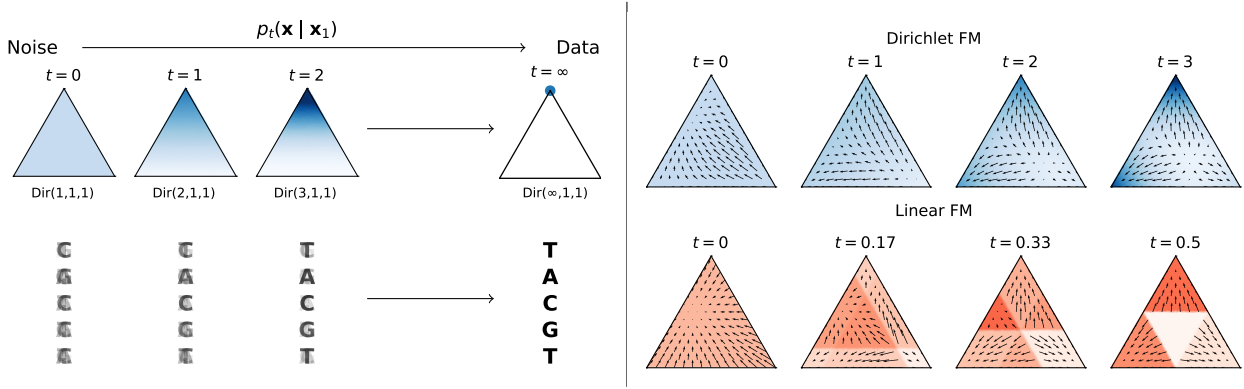
\begin{figure*}[ht!]
    \centering
    %\missingfigure[figwidth=\textwidth, figheight=0.3\textwidth]{Graphical abstract}
    \begin{tikzpicture}
    \node[anchor=south west,inner sep=0] (image) at (0,0) {\includegraphics[width=0.49\textwidth]{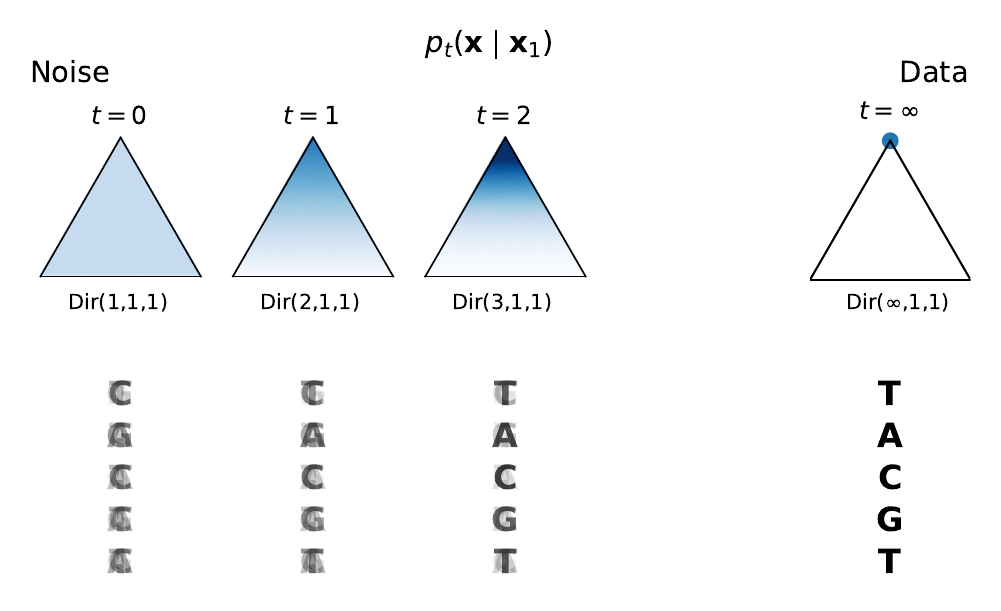}};
    \begin{scope}[x={(image.south east)},y={(image.north west)}]
        \draw[black] [->] (0.15,0.88) -- (0.85,0.88);
        \draw[black] [->] (0.6,0.67) -- (0.8,0.67);
        \draw[black] [->] (0.6,0.23) -- (0.8,0.23);
    \end{scope}
    \end{tikzpicture}
    \vrule
    \includegraphics[width=0.49\textwidth]{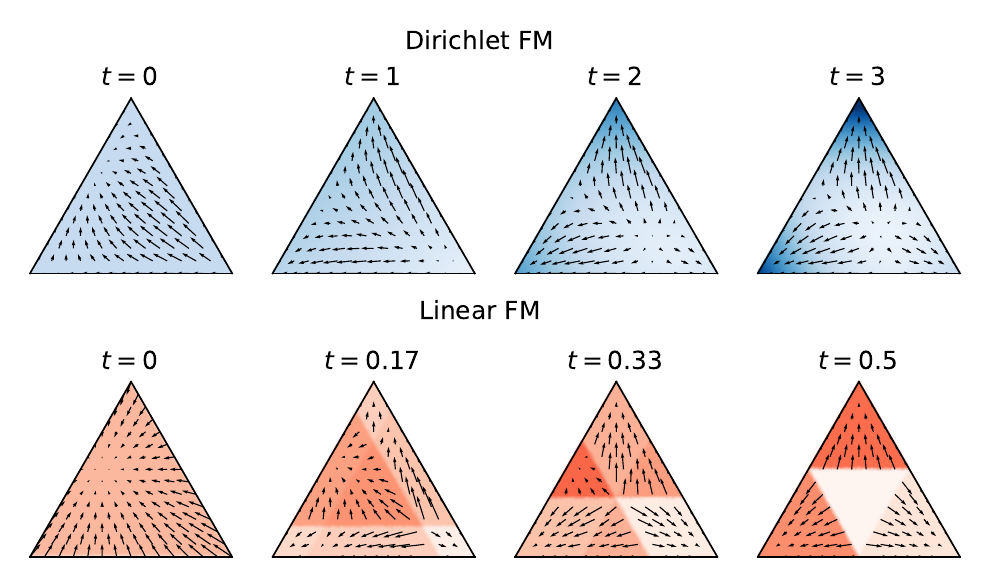}
    \vspace{-0.3cm}
    \caption{\textbf{Overview of Dirichlet flow matching}. We represent a sequence of discrete variables as a sequence of simplices. Here, we only show the simplex of one of the sequence positions. \emph{Left:}  starting from uniform noise on the probability simplex, we define conditional probability paths that approach a point mass at the vertex via a one-parameter family of Dirichlet distributions. We view a sequence of tokens as a sequence of simplices for which the probability path corresponds to noising tokens via \emph{superposition} with all other possible tokens {(during inference, the simplices depend on each other through a joint denoiser)}. \emph{Right:} Comparison of the marginal probability paths and vector fields of Dirichlet and linear FM. The vector fields of Dirichlet FM are smooth in time and space, unlike linear FM.}
    \label{fig:figure1}
\end{figure*}

% does not require sampling from intractable transition kernels at training-time.

% is technically valid but
% significantly easier; however, we show that  and carefully engineer Dirichlet FM to avoid these shortcomings.
%  probability path, in which the vector field linearly transports noise to data.
% 

Dirichlet FM is equipped with several features that showcase its potential as a generative modeling framework. A key aspect of modern frameworks such as diffusion is the ability to implement \emph{guidance} towards a target class or attribute \citep{dhariwal2021diffusion, ho2022classifier}. While flow matching guidance has been explored for (diffusion-like) Gaussian noising processes in $\mathbb{R}^n$ \citep{dao2023flow, zheng2023guided}, no general theory of flow matching guidance is available. However, by deriving a relationship between the score of a mixture of Dirichlets and the resulting flow, we implement both classifier guidance and classifier-free guidance within the framework of Dirichlet FM. Additionally, unlike autoregressive models or diffusion models that inject categorical noise \citep{austin2021structured, hoogeboom2021autoregressive, campbell2022continuous}, the sampling process of Dirichlet FM defines a deterministic map between noise and data. Hence, we can \emph{distill} the sampling process and generate arbitrarily long sequences in a single forward pass.

We evaluate Dirichlet FM on three DNA sequence datasets with several complex tasks that pose diverse challenges. First, we demonstrate that Dirichlet FM better generates promoter DNA sequences with desired regulatory activity compared to baselines such as an autoregressive language model or discrete diffusion models, including DDSM \citep{avdeyev2023dirichlet}. Second, we show on two enhancer DNA sequence datasets that Dirichlet FM improves upon autoregressive language modeling in capturing the data distribution with an FBD (distributional similarity) of 1.9 vs. 36.0 in Melanoma DNA and 1.0 vs. 25.2 in Fly Brain DNA. Third, we demonstrate that Dirichlet FM guidance can improve unconditional sequence generation and generate cell-type specific enhancer DNA sequences that improve upon the (experimentally validated) sequences of \citet{Taskiran2023cell}. Lastly, distilled Dirichlet FM generates sequences in a single step, resulting in orders-of-magnitude speedups with minimal performance degradation. % In all cases, Dirichlet FM significantly improves upon na\"ive linear FM.

\section{Background}

\subsection{Flow Matching}\label{sec:fm_background}

In flow matching \citep{lipman2022flow, liu2022flow, albergo2022building}, we consider a noisy distribution $\bfx_0 \sim q_0$ and data distribution $\bfx_1 \sim p_\text{data}$ and regress a neural network against a vector field that transports $q(\bfx_0)$ to $p_\text{data}(\bfx_1)$. To do so, we define a \emph{conditional probability path}---a time-evolving distribution $p_t(\bfx \mid \bfx_1), t\in[0,1]$ conditioned on $\bfx_1$ with boundary conditions $p_0(\bfx \mid \bfx_1) = q(\bfx)$ and $p_1(\bfx \mid \bfx_1) \approx \delta(\bfx - \bfx_1)$. We additionally assume knowledge of a \emph{conditional vector field} $u_t(\bfx~\mid~\bfx_1)$ that generates $p_t(\bfx \mid \bfx_1)$, i.e. satisfying the transport equation
\begin{equation}\label{eq:continuity-equation}
    \del{}{t}p_t + \nabla \cdot (p_tu_t) = 0
\end{equation}
Then, the marginal probability path
\begin{equation}
    p_t(\bfx) = \int p_t(\bfx \mid \bfx_1) p_\text{data}(\bfx_1)\, d\bfx_1 
\end{equation}
interpolates between noise $p_0 = q_0$ and data $p_1 \approx p_\text{data}$ and is generated by the \emph{marginal vector field}
\begin{equation}
    v_t(\bfx) = \int u_t(\bfx\mid \bfx_1) \frac{p_t(\bfx \mid \bfx_1) p_\text{data}(\bfx_1)}{p_t(\bfx)}\, d\bfx_1
\end{equation}
Thus, by learning and integrating a neural network $\hat v(\bfx, t; \theta) \approx v_t(\bfx)$, we can generate data from noisy samples $\bfx_0 \sim q$. The core design decision is the choice of appropriate conditional probability path $p_t(\bfx \mid \bfx_1)$ and associated vector field $u_t(\bfx \mid \bfx_1)$. Although it is possible to define these directly, it is often simpler to instead define a \emph{conditional flow map} $\psi_t(\bfx_0 \mid \bfx_1)$ that directly transports $\bfx_0 \sim q$ to the intermediate distribution $p_t(\bfx \mid \bfx_1)$. The flow map immediately provides the corresponding vector field:
\begin{align}
    u_t(\bfx \mid \bfx_1) = \frac{d}{dt}\psi_t(\bfx_0 \mid \bfx_1)
\end{align}
With this formulation, the required boundary conditions simplify to $\psi_0(\bfx_0 \mid \bfx_1) = \bfx$ and $\psi_1(\bfx_0 \mid \bfx_1) = \bfx_1$. As advocated by several works \citep{liu2022flow, lipman2022flow, pooladian2023multisample, tong2023improving}, the flow map (also called \emph{interpolant}) is often chosen to follow the simplest possible path between the two endpoints---e.g., linear in Euclidean spaces and geodesic on Riemannian manifolds \citep{chen2023riemannian}.

\subsection{Related Discrete Diffusion Model Works} \label{sec:discrete-diffusion}

% Discrete diffusion promises an edge over autoregressive language models in faster and more controllable generation. 
% For the purposes of distinguishing Dirichlet FM from DDSM, we should emphasize that in DDSM, only the \emph{stationary distribution} is a Dirichlet, while the marginals of their diffusion process are not. Meanwhile, Dirichlet FM has \emph{probability paths} that are mixtures of Dirichlets.
% Simplex-based approaches preserve the discrete structure of the data. Meanwhile, t
%We note that although some of these simpelx diffusions (e.g., DDSM) target a Dirichlet distribution as the noisy prior, none of them feature Dirichlet distributions

Existing discrete diffusion frameworks can be split into 4 categories. Firstly, simplex-based approaches frame discrete data as vertices of a simplex and generate it starting from a Dirichlet prior over the whole simplex \citep{richemond2022categorical, floto2023diffusion}. Among those, {DDSM} \citep{avdeyev2023dirichlet} is most related to our work and converges to a Dirichlet distribution via Jacobi diffusion processes and the stick-breaking transform. We note that none of these simplex-based approaches feature Dirichlet distributions as intermediate distributions of the noising process---a key aspect of our approach. 

The second class of discrete diffusion models fully relaxes discrete data into continuous space without any constraints and uses, e.g., a standard Gaussian as prior \citep{han2022ssd, chen2023analog, frey2024protein}. The third paradigm, established by D3PM \citep{austin2021structured}, operates on discrete samples of noise distributions constructed by injecting discrete noise into data \citep{campbell2022continuous, igashov2023retrobridge, vignac2023digress, penzar2023legnet}. Lastly, \emph{latent} discrete diffusion models train an additional network to obtain continuous latents for which they train a conventional diffusion model \citep{dieleman2022continuous, li2023latent}. 

% In this framework, it is unclear how to perform guidance or distillation, which is possible with Dirichlet FM. 

%, such as regulating which genes are translated.

% A larger portion of DNA is not transcribed but still fulfills important functions 
% %An enhancer could, for instance, function by binding a transcription factor that then recruits RNA polymerase \uppercase\expandafter{\romannumeral2\relax}, which starts transcribing the nearby gene \citep{koch2011transcription, panigrahi2021mechanisms}. 
\subsection{Promoter and Enhancer DNA} \label{sec:dna-background}
DNA is a sequence with \emph{base pairs} as tokens (3 billion for humans) and a vocabulary of 4 \emph{nucleotides} (A, T, C, G). Parts of DNA encode genes that are transcribed into mRNA and then translated to functional proteins. \emph{Promoters} and \emph{enhancers} refer to noncoding portions of DNA that regulate the expression level of these genes and play important roles in eukaryotic organisms such as humans \citep{dunham2012encode, luo2020encode}. More specifically, a promoter for a gene is the DNA sequence next to the gene where the transcriptional machinery binds and starts transcribing DNA to mRNA \citep{haberle2018eukaryotic}. Meanwhile, enhancers are sequences that can be distant in the DNA sequence (millions of base pairs) but are close in 3D space \citep{panigrahi2021mechanisms} and regulate the recruitment of this transcriptional machinery. Unlike promoters, enhancers often regulate transcription in specific cell types. Hence, while both types of DNA subsequences are important for gene therapy \citep{whalen1994genetherapeutics}, the cell type specificity of enhancers enables targeting, e.g., only cancer cells. 

\textbf{Designing enhancers.} Recently \citet{Taskiran2023cell} and \citet{deAlmeida2023targeted} demonstrated cell type-specific enhancer design via an \emph{optimization} procedure starting from an initial sequence guided by a cell-type activity classifier. However, such sequence designs may not follow the empirical distribution of enhancers, which would be captured by a generative model. Our work sets the foundation for more principled conditional sequence design by learning and drawing from conditional data distributions.

% By demonstrating class-conditional enhancer design with a generaThus, a cell-type conditioned generative model that captures the enhancer distribution would be valuable by solving the shortcomings of optimization approaches. 

% , meaning that they may fail to satisfy typical enhancer properties. 
% may be limited in their diversity and 

\section{Method}

\subsection{Flow Matching on the Simplex}

Let $S_K$ be the probability simplex in $K$-dimensional space:
\begin{equation}
    S_K = \{\bfx = (x_1, \ldots x_K)^T \in \mathbb{R}^K \mid \boldsymbol{1}^T\bfx = 1, \bfx \ge 0\}
\end{equation}
Given a $K$-class categorical distribution with probabilities $p_i, \sum_{i=1}^K p_i = 1$, we relax this distribution into continuous space by converting it to a mixture of point masses at the vertices of $S_k$ (with $\bfe_i$ as the $i$th one-hot vector):
\begin{equation}
    p_\text{data}(\bfx) = \sum_{i=1}^K p_i\delta(\bfx - \bfe_i)
\end{equation}
We then define the noisy prior to be the uniform density on the simplex, or a \emph{Dirichlet distribution} with parameter vector $\boldsymbol{\alpha}$ given by the all ones vector:
\begin{equation}\label{eq:uniform-dirichlet}
    q_0(\bfx) = \Dir(\bfx; \boldsymbol{\alpha} = (1, \ldots 1)^T) = \Gamma(K) % \frac{1}{\sqrt{K}}
\end{equation}
Our objective is then to learn a vector field, using some choice of conditional probability path (discussed later), to transport $q_0$ to $p_\text{data}$. Typically, the neural network directly parameterizes the vector field and is trained via the $L_2$-like conditional flow-matching loss
\begin{equation}
    \mathcal{L}(\theta) = \mathbb{E}[\lVert u_t(\bfx \mid \bfx_1) - \hat v(\bfx, t; \theta) \rVert^2] 
\end{equation}
where the expectation is taken over $t \sim \mathcal{U}[0,1], \bfx_1 \sim p_\text{data}, \bfx \sim p_t(\cdot \mid \bfx_1)$. However, we instead train a \emph{denoising classifier} via a cross-entropy loss
\begin{equation}
    \mathcal{L}(\theta) = -\mathbb{E}[\log \hat p(\bfx_1 \mid \bfx; \theta)] 
\end{equation}
At inference time, we then parameterize the vector field via
\begin{equation} \label{eq:marginal-vectorfield}
    \hat v(\bfx, t;\theta) = \sum_{i=1}^K u_t(\bfx \mid \bfx_1 = \bfe_i) \hat p(\bfx_1 = \bfe_i \mid \bfx; \theta)
\end{equation}
It can be shown (Appendix \ref{appx:method-details}) that the two losses have the same minimizer, and thus, the cross-entropy is a valid flow-matching objective. The advantages of this approach are twofold: (1) it ensures that the learned vector field is restricted to the tangent plane of the simplex (i.e., the components sum to zero), and (2) the conditional vector field does not need to be evaluated at training time.

For simplicity, our discussion focuses on modeling categorical data on the simplex. However, in practice, we are interested in \emph{sequences} of variables relaxed onto the \emph{multi-simplex} $S_K^N$. At inference time, the simplices depend on each other through a learned denoiser that outputs token-wise logits conditioned on all noisy inputs. This extension to a sequence of categorical tokens is analogous to the development in previous works, such as D3PM (\citep{austin2021structured} pg 3) and DDSM \citep{avdeyev2023dirichlet}. In particular, noise is added to each token independently to produce a noisy sequence, and the model is trained to predict the denoised token at each site. The model makes this prediction using information about all noisy tokens, thus injecting dependencies between positions into the generative process. The flow field is then constructed independently at each sequence position from the corresponding logits.

%\textcolor{blue}{ At inference time, the couplings between tokens are enforced by the learned noisy classifier, which outputs token-wise logits conditioned on all noisy inputs.}

%In this case, the probability path on the multi-simplex is defined to be the product of the appropriate conditional probability paths---that is, the noising processes are independent.

%all statements and operations apply token-wise as necessary. In other words, similar to diffusion processes over multiple coordinates, the forward process $p($

\subsection{Designing Simplex Flow Matching} \label{sec:linear-fm}

\begin{figure}
    \centering
    \includegraphics[width=\columnwidth]{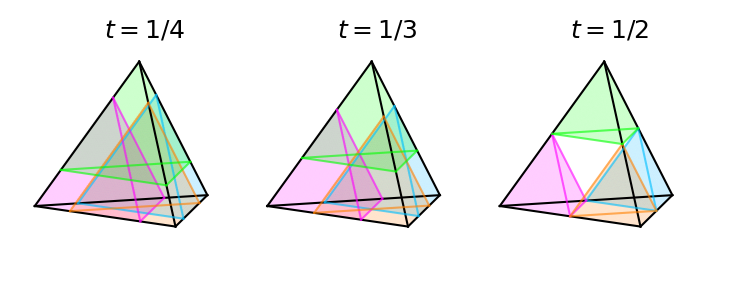}
    %\vspace{-20pt}\hrule \vspace{10pt}
    %\vspace{-10pt}
    \includegraphics[width=0.85\columnwidth]{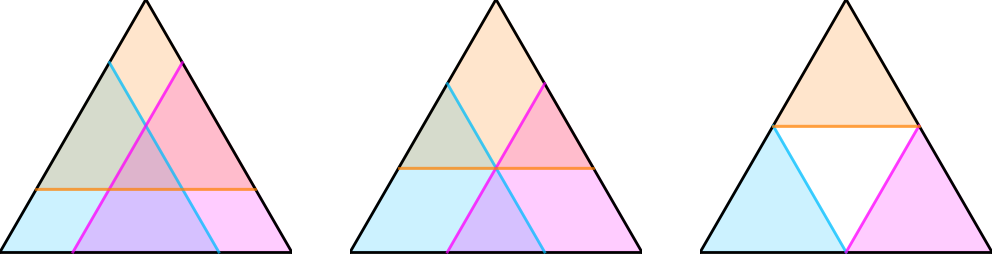}
    \vspace{-0.2cm}\caption{\textbf{Pathological behavior of linear flow matching} on the simplex with $K=4$ (\emph{top}) and $K=3$ (\emph{bottom}). Each color represents a conditional probability path evolving over time toward its target vertex. At $t=1/4$, $t=1/3$, and $t=1/2$, the region of overlap between 4, 3, and 2 conditional probability paths disappears, respectively, corresponding to a shrinking set of possible values of $\bfx_1 \mid \bfx$ for any $\bfx$.} %\emph{Bottom}: the discontinuous density of the conditional probability path causes discontinuities in the vector field across time and space.
    \label{fig:linear-probability-path}
\end{figure}
As outlined in Section~\ref{sec:fm_background}, there are two options to define a conditional probability path $p_t(\bfx \mid \bfx_1)$ and corresponding vector field $u_t(\bfx \mid \bfx_1)$ to train a flow model with:
\begin{enumerate}
    \item \textbf{Interpolant perspective:} Define an {interpolant} $\psi_t(\bfx_0 \mid \bfx_1)$, which provides the density $p_t(\bfx \mid \bfx_1)$ implicitly but allows one to easily sample from it, and obtain the conditional vector field {trivially} by taking the derivative $u_t = \partial_t \psi_t$.
    \item \textbf{Probability path perspective:} Define $p_t(\bfx \mid \bfx_1)$ explicitly and solve for $u_t(\bfx \mid \bfx_1)$ that satisfies the transport equation which can be {non-trivial}.
\end{enumerate}
Following extant works on flow matching, the most natural way to proceed for the simplex would be to follow the \emph{interpolant perspective} and use the linear flow map employed in \citet{lipman2022flow, pooladian2023multisample}:
\begin{align}
    \psi_t(\bfx_0 \mid \bfx_1) &= (1-t)\bfx_0 + t\bfx_1\label{eq:linear-flow-map}\\
    u_t(\bfx \mid \bfx_1) &= \frac{\bfx_1 - \bfx}{1-t} = \bfx_1 - \bfx_0 \label{eq:linear-vf}
\end{align}
Since $S_K$ is a Euclidean space, these operations remain well-defined, and the interpolant transports all points on the simplex to $\bfx_1$ at $t=1$ via straight paths. However, a pathological property emerges when such conditional probability paths are marginalized over $p_\text{data}(\bfx_1)$ over the course of flow matching training:

\begin{proposition}
    Suppose that a flow matching model is trained with the linear flow map (Equation \ref{eq:linear-flow-map}). Then, for all $k = 2, \ldots K$ and $\bfx \sim p_t(\bfx)$, the converged model posterior $p(\bfx_1 \mid \bfx) \propto p_t(\bfx \mid \bfx_1)p_\text{data}(\bfx_1)$ has support over at most $k-1$ vertices for times $t > 1/k$.
\end{proposition}

Conceptually, this means that as the model transports samples on the simplex at $t=0$ to the vertices of the simplex at $t=1$, it must eliminate or rule out a possible destination vertex at each of the times $\frac{1}{K}, \frac{1}{K-1}, \ldots \frac{1}{2}$, if not earlier. As $K$ becomes large, an increasingly large fraction of the model capacity must be allocated to a smaller and smaller fraction of the total time and trajectory length---indeed, for all $K$, the posterior for times $t > 1/2$ reduces to the $\arg\max$ operator. Further, since the marginal field is increasingly discontinuous (i.e., rapidly changing directions and settings entries to zero corresponding to eliminated vertices), the model becomes increasingly sensitive to integration step size. We posit---and empirically verify in Section \ref{sec:experiments}---that these factors significantly hurt the performance of linear flow matching, especially for higher dimensionalities $K$. 

Upon examination of the conditional probability paths (Figure~\ref{fig:linear-probability-path}), it becomes clear that this pathological behavior is due to the shrinking support of the conditional paths that arise from linear FM. Unfortunately, by following the \emph{interpolant perspective}, we are unable to directly control the conditional path. Hence, to obtain a method that does not suffer from linear FM's pathologies, we proceed with the \emph{probability path perspective} and directly define the conditional probability path so that it has support on the entire simplex {at all times}, as described next.

\subsection{Dirichlet Flow Matching}

\paragraph{Probability path $p_t(\bfx \mid \bfx_1)$.} Recall the \emph{Dirichlet distribution's} probability density function (where $\mathcal{B}$ is the beta function as in Equation \ref{eq:beta-function}):
\begin{equation} \
    \Dir(\bfx; \alpha_1, \ldots \alpha_K) = \frac{1}{\mathcal{B}(\alpha_1, \ldots \alpha_K)}\prod_{i=1}^K x_i^{\alpha_i-1}
\end{equation}
Following the \emph{probability path perspective}, we first define a conditional probability path with $t \in [0, \infty)$ as:
\begin{equation}\label{eq:probability-path}
    p_t(\bfx \mid \bfx_1 = \bfe_i) = \Dir(\bfx; \boldsymbol{\alpha} = \boldsymbol{1} + t\cdot\bfe_i)
\end{equation}
% Indeed, there are an infinite number of probability fluxes $p_tu_t$---equivalent up to a divergence-free component---which satisfy the transport Equation .
When $t=0$, this is equal to the uniform noise distribution (Equation \ref{eq:uniform-dirichlet}). As $t \rightarrow \infty$, the $i$th entry of the $\boldsymbol{\alpha}$ parameter vector increases while the others remain constant, concentrating the density towards a point mass on the $i$th vertex, corresponding to the $t=1$ boundary condition in standard flow matching.\footnote{We continue to call the data sample $\bfx_1$, and in practice, integrate to some large fixed time (typically $t=8$) and take the $\arg\max$ of the final model posterior.} Hence, this family of Dirichlet distributions provides a conditional probability path with the required boundary conditions while retaining support over the entire simplex, as desired.

\paragraph{Vector field $u_t(\bfx \mid \bfx_1)$.} Since we have chosen a conditional probability path directly rather than implicitly via an interpolant, it is more difficult to obtain the corresponding conditional vector field $u_t(\bfx \mid \bfx_1)$. Indeed, there is an infinite number of such fields that generate the desired evolution of $p_t$. Motivated by the basic form of the linear FM, we generalize it via the following \emph{ansatz}:
\begin{equation}\label{eq:vectorfield}
    u_t(\bfx \mid \bfx_1 = \bfe_i) = C(x_i, t) (\bfe_i -\bfx)
\end{equation}
That is, the flow still points directly towards the target vertex $\bfe_i$, but is rescaled by a $x_i$-dependent factor. The conditional vector field of linear FM satisfies this form with $C(x_i, t)=1/(1-t)$ dependent only on $t$; we introduce the additional $x_i$-dependence to control the contraction of probability mass towards $\bfe_i$. 
In Appendix \ref{appx:vf-derivation}, we derive the $C(x_i, t)$, which gives rise to Dirichlet probability paths to be:
% \begin{multline}
%     u_t(\bfx \mid \bfx_1 = \bfe_i) = \\ - \tilde I_{x_i}(t+1, K-1)\frac{\mathcal{B}(t+1, K-1)}{(1-x_i)^{K-1} x_i^t}(\bfe_i - \bfx)
% \end{multline}
\begin{equation}
    C(x_i, t) =  - \tilde I_{x_i}(t+1, K-1)\frac{\mathcal{B}(t+1, K-1)}{(1-x_i)^{K-1} x_i^t}
\end{equation}
where
\begin{equation}
    \tilde I_x(a, b) = \del{}{a}I_x(a, b)
\end{equation}
%final vector field is non trivial and not easy to come up with without this perspective.
% It has a spatial dependence in the magnitude and it is not at all clear what this should be. This is what makes it hard. The vector field that does not have the issues needs a spatial dependence like this.
is a derivative of the \emph{regularized incomplete beta function} $I_x(a, b)$.
\begin{proposition}
The proposed vector field (Eq \ref{eq:vectorfield}) and probability path (Eq \ref{eq:probability-path}) together satisfy the transport equation (Eq. \ref{eq:continuity-equation}) and therefore constitute a valid flow-matching framework. (Proof in Appendix \ref{appx:vf-derivation}).
\end{proposition}

Figure~\ref{fig:vf-magnitude} compares the magnitude of the resulting vector field with that of linear FM as a function of distance from the target vertex. As anticipated, the field vanishes both at the target vertex \emph{and} on the $(K-2)$-dimensional face directly opposite it instead of monotonically scaling with distance from the target vertex as in linear FM. This second condition means the probability density is never fully drawn away from the face and resolves the pathological behavior of linear FM. The resulting probability paths and vector fields are visualized on the simplex in Figure \ref{fig:figure1}; they are continuous and smooth, unlike in linear FM.

% the magnitude has more complex spatial dependence:

\begin{figure}
    \centering
    \includegraphics[width=\columnwidth]{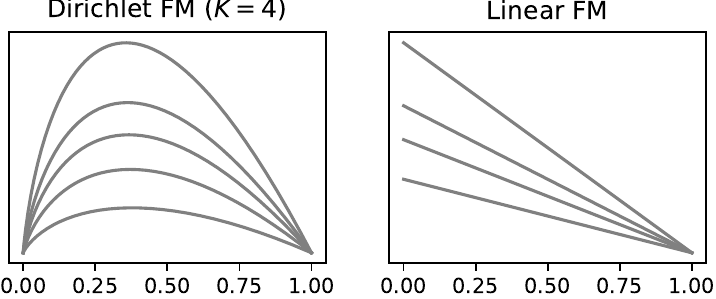}
    \vspace{-16pt}
    \caption{\textbf{Vector field magnitudes} of the conditional flow field $u_t(\bfx \mid \bfx_1 = \bfe_i)$ as a function of $x_i$ (the $i$th element of $\bfx$) plotted for varying values of $t$. In Dirichlet FM, the field vanishes at both $x_i=1$ (i.e., the target vertex) and $x_i=0$ (the opposite face).}
    \label{fig:vf-magnitude}
\end{figure}

\subsection{Guidance} \label{sec:method-guidance}
A key attribute of iterative generative models is the ability to continuously and gradually bias the generative process towards a class label with user-specified strength, a technique known as \emph{guidance} \citep{dhariwal2021diffusion, ho2022classifier}. Initially proposed in the context of diffusion models, where the generative process follows the \emph{score} $\hat s(\bfx, t;\theta) \approx \nabla_\bfx \log p_t(\bfx)$ of the noisy data distribution, guidance is implemented by taking a linear combination of the unconditional and conditional score models
\begin{equation}
    \hat s_\text{CFG}(\bfx, t, y;\theta) = \gamma \hat s(\bfx, t, y; \theta) + (1 - \gamma)\hat s(\bfx, t, \varnothing; \theta)
\end{equation} % is equal to the score of the true class-conditioned noisy marginal $p_t(\bfx\mid y)$ and
with $\gamma > 0$ and running the generative process with this adjusted score. In the context of flow matching, \citet{dao2023flow, zheng2023guided} derived a relationship between the score and marginal flow for certain Gaussian probability paths and implemented flow guidance by propagating the effects of standard score adjustments to the resulting flow fields. We implement guidance for Dirichlet FM by deriving a similar relationship between the score and flow field, detailed below. Note that when $\gamma = 1$, this adjustment precisely mimics the flow that would result from  training only on points with class label $y$; however, similar to prior works, we find $\gamma > 1$ futher enhances the guidance efficacy.

% When $\gamma=1$, the resulting score correctly mimics the score if the model had been trained only on samples with label $y$.
% ;  Empirically, $\gamma > 1$ further enhances the guidance efficiency. 
% When $\gamma=1$, this adjustment can be derived from Bayes' rule of the noisy marginal distribution used in the generative process. In classifier guidance, this adjustment corresponds to multiplying $p_t(\bfx)$ with the class probability from a noisy classifier $p_t(y \mid \bfx)$ raised to exponent $\gamma$; in classifier-free guidance, this classifier is imputed based on access to a conditional score model $s_t(\bfx \mid y)$. 
% , we similarly perform guidance by adjusting scores corresponding to multiplication with a noisy classifier and derive the corresponding effects on the marginal flow.

% , although the resulting score may not correspond to the desired (or any valid) noisy marginal
% Although the meaning of this transformation is in general unclear when the score no longer corresponds to a valid noisy marginal,

\textbf{Relationship between flow and score.} For the Dirichlet marginal probability path, the score can be obtained from the model posterior via the denoising score-matching identity \citep{song2019generative}:
\begin{equation}
    \hat s(\bfx, t; \theta) = \sum_{i=1}^K s_t(\bfx \mid \bfx_1 = \bfe_i)\hat p(\bfx_1 = \bfe_i \mid \bfx; \theta)
\end{equation}
We can differentiate Equation \ref{eq:probability-path} to obtain a matrix equation
\begin{equation}
    \hat s = \mathbf{D}\hat p, \quad \text{where} \quad \mathbf{D}_{ij} = \delta_{ij}\frac{t}{x_i}
\end{equation}
Here, $\mathbf{D}$ is a $K \times K$ diagonal matrix dependent on $\bfx, t$ and $\hat s, \hat p \in \mathbb{R}^n$. (Technically, $\hat s$ contains both on-simplex and off-simplex components, the latter of which is ignored.) Meanwhile, the computation of the marginal flow (Equation \ref{eq:marginal-vectorfield}) can also be written as a very similar matrix equation $\hat v = \mathbf{U}\hat p$ where the entries of $\mathbf{U}$ are given by Equation 15. Combining these, we obtain
\begin{equation}\label{eq:score-to-vf}
    \hat v = \mathbf{U}\mathbf{D}^{-1}\hat s
\end{equation}
where $\mathbf{D}$ is invertible since it is diagonal with nonnegative entries. Thus, a \emph{linear relationship} exists between the marginal flow and the score arising from the same model posterior.

% In common practice, the noisy classifier used in Equation  is not actually trained, but rather imputed from the difference of the condition and unconditional score models \citep{ho2022classifier}:
% \begin{equation}
%     \hat s_\text{CFG}(\bfx, t, y;\theta) = \gamma \hat s(\bfx, t, y; \theta) + (1 - \gamma)\hat s(\bfx, t, \varnothing; \theta)
% \end{equation}
% Now 
\textbf{Classifier-free guidance.} Suppose we have class-conditional and unconditional flow models $\hat v(\bfx, t, y; \theta)$ and $\hat v(\bfx, t, \varnothing; \theta)$. Since a linear combination of scores results in a linear combination of flows, we similarly implement guidance by integrating 
\begin{equation}\label{eq:cfg}
    \hat v_\text{CFG}(\bfx, t, y;\theta) = \gamma \hat v(\bfx, t, y; \theta) + (1 - \gamma)\hat v(\bfx, t, \varnothing; \theta).
\end{equation}
In practice, we linearly combine the probabilities that correspond to the conditional and unconditional vector fields through the derived linear relationship. When using $\gamma > 1$, this can result in negative values in $\hat p$ and we project the vector back onto the simplex via the algorithm of \citet{wang2013projection}.

\textbf{Classifier guidance.} In cases where a conditional flow model is unavailable, we use the gradient of a noisy classifier to obtain a conditional score from an unconditional score:
\begin{equation}\label{eq:classifier-guidance}
    \hat s(\bfx, t, y; \theta) = \hat s(\bfx, t, \varnothing; \theta) + \nabla_\bfx \log \hat p(y \mid \bfx, t; \theta)
\end{equation}
The conditional scores can then be converted into a model posterior $\hat p$ and then a marginal flow $\hat v$ via Equation \ref{eq:score-to-vf}. However, a direct application is not possible because the classifier gradients do not have the appropriate off-simplex components to ensure a valid model posterior (i.e., $\hat p \in S_K$) when operated on by $\mathbf{D}^{-1}$. Instead, we modify Equation \ref{eq:score-to-vf} via
\begin{equation}
    \tilde {\mathbf{D}} = \left(I - \frac{1}{K}\boldsymbol{1}\boldsymbol{1}^T\right)\mathbf{D}
\end{equation}
corresponding to projecting the score onto the tangent plane of the simplex (i.e., $\boldsymbol{1}^T\hat s = 0$). As now $\mathbf{D}$ is no longer invertible, we obtain the conditional model posterior from the score by solving $\hat s = \tilde{\mathbf{D}}\hat p$ with the additional constraint $\boldsymbol{1}^T\hat p = 1$. This can result in negative values in $\hat p$ and we project to the simplex via the algorithm of \citet{wang2013projection}.

\subsection{Distillation}
The aim of \emph{distillation} \citep{salimans2022progressive, song23consistency, yin2023onestep} is to reduce the inference time of the iterative generative process by reducing the number of steps while retaining sample quality. However, for discrete diffusion models (see Section \ref{sec:discrete-diffusion}) or autoregressive language models, no distillation techniques exist, and it is unclear how to distill generative models based on discrete noise. For Dirichlet FM, inference is a deterministic ODE integration defining a map between the prior and target distribution. Hence, we can employ standard distillation approaches and distill the teacher model (using 100 steps in our experiments) into a student model representing the map. For this, we sample the teacher to obtain pairs of noise and training targets. We use these to train the student model to reproduce the teacher distribution in a single step. With this, we are able to reduce inference times and provide the first demonstration of distillation for flow matching and for iterative generative models of discrete data.

\section{Experiments} \label{sec:experiments}

\begin{figure}
    \centering
    \includegraphics[width=0.4\textwidth]{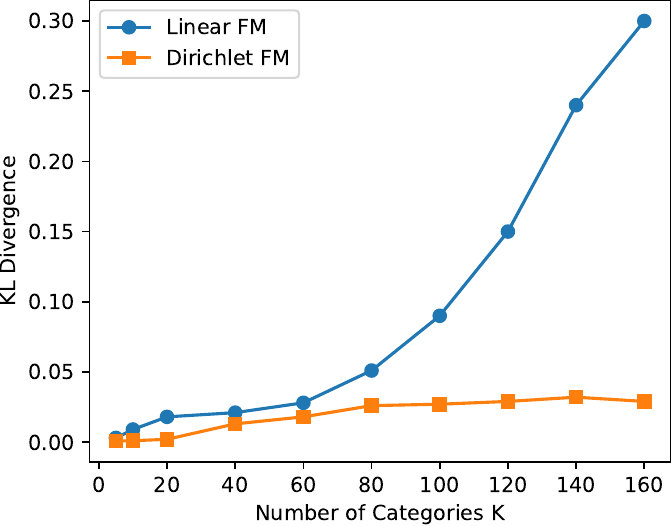}
    %\vspace{-0.5\baselineskip}
    \caption{\textbf{Scaling to higher simplex dimensions.} We train on  simple categorical distributions with an increasing number of categories $K$ and measure the KL divergence between the generated distributions (512k samples) and the training target distribution. Dirichlet FM scales to larger $K$ much better than linear FM.}
    \label{fig:toy-experiment}
    %\vspace{-\baselineskip}
\end{figure}
\subsection{Simplex dimension toy experiment} \label{sec:toy-experiment}
We first evaluate Dirichlet FM and linear FM in a simple toy experiment where the KL divergence of the generated distribution to the target distribution can be evaluated under increasing simplex dimensions.  For this, we train both methods to reproduce a categorical distribution $q(\bfx)=\Cat(K, \mathbf{a})$ where the class probabilities $\mathbf{a}$ are sampled from a uniform Dirichlet  $\mathbf{a} \sim \Dir(\mathbf{a}; \boldsymbol{1})$. To evaluate, we use the KL-divergence KL$(\Tilde{p} \parallel q)$ where $\Tilde{p}$ is the empirical distribution of $512,000$ model samples. 

Figure \ref{fig:toy-experiment} shows these KL-divergences for increasing sizes $K$ of the categorical distribution.  Dirichlet FM is able to overfit on the simple distribution regardless of $K$. Meanwhile, Linear FM is unable to overfit on simple categorical distributions as $K$ increases, illustrating the practical implications of the pathological probability paths and discontinuous vector fields as discussed in Section \ref{sec:linear-fm}.%  are highly problematic in practice and not only in theory. % Thus, a better solution, such as Dirichlet FM, is necessary.

\subsection{Promoter DNA sequence design} \label{sec:promoter-experiments}
We next assess the ability of Dirichlet FM to design DNA promoter sequences conditioned on a desired \emph{promoter profile}. The experimental setup and evaluation closely follow that of DDSM \citep{avdeyev2023dirichlet}.

% (see Section \ref{sec:dna-background} for background on promoters) 
% profile with a scalar for each base pair, indicating the likelihood of transcription starting there 
\textbf{Data.} We use a dataset of $100,000$ promoter sequences with $1,024$ base pairs extracted from a database of human promoters \citep{hon2017atlas}. Each sequence has a CAGE signal \citep{shiraki2003cage} annotation available from the FANTOM5 promoter atlas \citep{forrest2014promoter}, which indicates the likelihood of transcription initiation at each base pair ($\bfr \in \mathbb{R}^{1024}$). Sequences from chromosomes 8 and 9 are used as a test set, and the rest for training. % to reduce train-test similarity.

% A profile $\bfr \in \mathbb{R}^{1024}$ consists of a scalar value for each of the 1024 base pairs that is related to the likelihood of transcription initiating there.
% Given a transcription profile, we aim to generate a promoter sequence with that profile.

\begin{table}[t]
\caption{\textbf{Evaluation of transcription profile conditioned promoter DNA sequence design.} Conditioned on a transcription profile, each method is tasked to generate a DNA sequence with that profile. The \textsc{MSE} is between the predicted regulatory activity of the designed sequence and the ground truth sequence. \textsc{NFE} refers to the number of function evaluations required for sampling. Numbers with * are from \citet{avdeyev2023dirichlet}.} \label{tab:promoter-design}
\begin{center}
\begin{small}
\begin{sc}
\begin{tabular}{lcc}
\toprule
Method & MSE & NFE\\
\midrule
Bit Diffusion (bit-encoding)*    & .0414 & 100\\
Bit Diffusion (one-hot encoding)* & .0395 & 100\\
D3PM-uniform*  & .0375 & 100\\
DDSM*  & .0334 & 100\\
Language Model & .0333 & 1024\\
\midrule
Linear FM & .0281 & 100\\
Dirichlet FM & \textbf{.0269} & 100\\
Dirichlet FM distilled & .0278 & \textbf{1}\\
\bottomrule
\end{tabular}
\end{sc}
\end{small}
\end{center}
%\vspace{-1.5\baselineskip}
\end{table}
%Promoters are an important type of DNA sequence that regulate whether or not Genes are transcribed. In biological research, a generative model for promoters could be of crucial value for discovering regulatory functions or designing promoters for gene therapies. Thus, we evaluate Dirichlet FM for promoter design following the setup of \citet{avdeyev2023dirichlet}. Please see Section \ref{sec:dna-background} for more background on regulatory DNA sequences.

\textbf{Task.} We train Dirichlet FM conditioned on a profile by providing it as additional input to the vector field. Following \citet{avdeyev2023dirichlet}, we evaluate generated sequences with the mean squared error (MSE) between their predicted regulatory activity and that of the original sequence corresponding to the input profile. The regulatory activity is determined by the promoter-related predictions of \textsc{Sei} \citep{chen2022sei}, a model trained on various regulatory signals.

\textbf{Baselines.} We compare Dirichlet FM with linear FM, discrete diffusion methods, and a language model that autoregressively generates the base pairs. The discrete diffusion baseline most related to our work is the simplex-based DDSM \cite{avdeyev2023dirichlet}. Our two other diffusion baselines are Bit Diffusion \citep{chen2023analog} and D3PM \citep{austin2021structured}. All methods use the same DNA modeling architecture and training protocol that was designed and tuned by \cite{avdeyev2023dirichlet} for DDSM. See Appendix \ref{appx:train-inference} for implementation details. 

%This architecture consists of 20 CNN layers with interspersed encodings of the diffusion or flow time. The diffusion and flow approaches employ 100 integration steps for inference, while Dirichlet FM distilled uses a single step, and the autoregressive language model requires $1,024$ steps.

\textbf{Results.} Dirichlet FM improves upon Linear FM, which are the only two methods that outperform the language model baseline (Table \ref{tab:promoter-design}).
The second best method in this comparison is the distilled version of Dirichlet FM, which retains almost the same performance. This means that our distilled Dirichlet FM outperforms all other methods in a single step, which is a $100\times$ speedup compared to the diffusion models and a $1,024\times$ speedup compared to the language model in terms of number of function evaluations (NFE).

\subsection{Enhancer DNA design}
%Similar to promoters, enhancers are DNA sequences that regulate (enhance) gene transcription. A key difference is their cell type specificity - commonly, they only enhance gene transcription in specific classes of cells \citep{panigrahi2021mechanisms}. With this, they are a promising tool for gene therapies that should only affect certain cells (e.g., cancer cells). Hence, we evaluate Dirichlet FM for generating enhancers and design cell-type specific enhancers with classifier and classifier-free guidance.

We now assess the performance of Dirichlet FM on DNA enhancer sequences and design evaluations that quantify both unconditional and conditional sample quality. Implementation and architecture details are in Appendix \ref{appx:train-inference}.

\textbf{Data.} We evaluate on two enhancer sequence datasets from fly brain cells \citep{janssens2022deepflybrain} and from human melanoma cells \citep{atak2021deepmel2}. These contain 104k fly brain and 89k melanoma sequences of length 500 with cell class labels determined from ATAC-seq data \citep{buenrostro2013atacseq}. Overall, there are 81 such classes of cells in the fly brain data and 47 in melanoma data (see Appendix \ref{appx:data} for more data details). 

\begin{table}[t]
\caption{\textbf{Evaluation of unconditional enhancer generation.} Each method generates 10k sequences, and we compare their empirical distributions with the data distribution using our Fr{\'e}chet Biological distance (FBD) analogous to FID for image generative models. \textsc{NFE} refers to number of function evaluations. The \textsc{Random Sequence} baseline shows the FBD for the same number and length of sequences with uniform randomly chosen nucleotides. \textsc{Dirichlet FM dist.} refers to our one-step distilled model and \textsc{Dirichlet FM CFG} to classifier-free guidance towards all classes with their empirical frequencies. } \label{tab:enhancer-design}
\begin{center}
\begin{small}
\begin{sc}
\begin{tabular}{lccc}
\toprule
     & Melanoma & Fly Brain& \\
Method & FBD & FBD   & NFE \\
\midrule
Random Sequence & 622.8 & 876 & -- \\
%DDSM & 71.8 & 84.2 & 100 \\
Language Model & 36.0 & 25.2 & 500 \\
Linear FM & 19.6 & 15.0 & 100\\
\midrule
Dirichlet FM & 5.3 & 15.2 & 100\\
Dirichlet FM dist. & 6.1 & 15.8 & \textbf{1}\\
Dirichlet FM CFG & \textbf{1.9} & \textbf{1.0} & 200\\
\bottomrule
\end{tabular}
\end{sc}
\end{small}
\end{center}
%\vspace{-1.5\baselineskip}
\end{table}
\textbf{Metric.} To score the similarity between a data distribution and a generative model's distribution, we employ a metric similar to the Fr{\'e}chet inception distance (FID) that is commonly used to evaluate image generative models \citep{heusel2017fid} and was adapted to molecule generative models as Fr{\'e}chet ChemNet distance (FCD) \citep{Preuer2018}. We follow this established principle and call our metric Fr{\'e}chet Biological distance (FBD). Hence, we train a classifier model to predict cell types and use its hidden representations as embeddings of generated samples and data distribution samples. Then, the FBD is calculated as the Wasserstein distance between Gaussians fit to embeddings from the two distributions (10k each). 
%  in each experiment for each method to obtain the 10k embeddings to represent the model distribution for the FBD calculation

%\textbf{Implementation.} All baselines and our Dirichlet FM use the same architecture, which is the CNN from our promoter experiments with minor modifications. The classifier for the FBD evaluation uses a similar network followed by a pooling layer and a classification head from which we extract the embeddings. Its difference from the generative model CNNs is the lack of time encoding and taking learned base pair encodings as input instead of points on the simplex. Finally, for classifier guidance, we train an additional time-dependent classifier that is used during inference. 

\begin{figure}[t]
    \centering
    \includegraphics[width=0.49\textwidth]{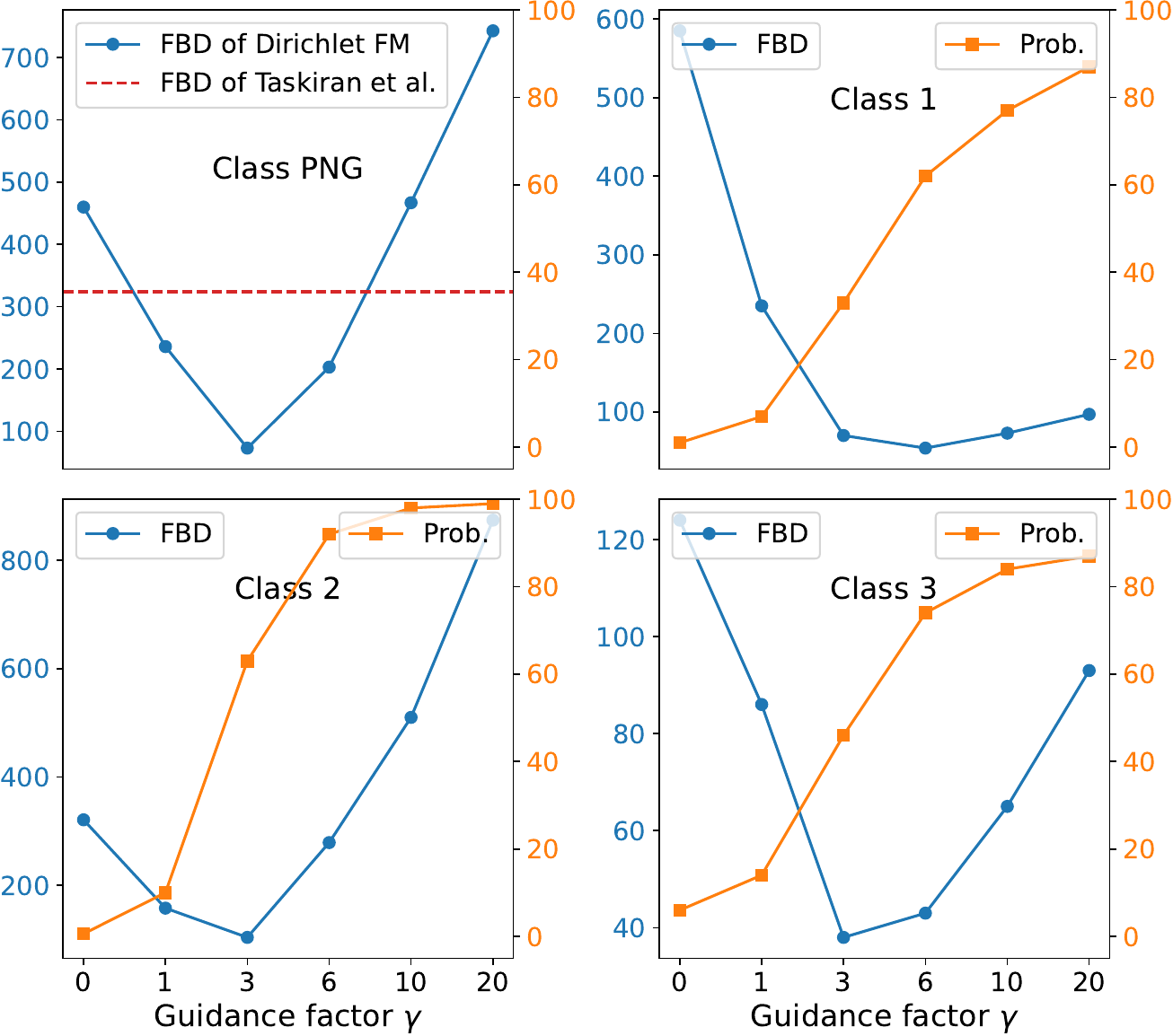}
    \vspace{-1.2\baselineskip}
    \caption{\textbf{Classifier-free guidance for cell type \emph{conditional} enhancer design.} We generate enhancers that are only active in cell class via classifier-free guidance with varying $\gamma$. Shown are 4 classes of the Fly Brain cell data. The \textcolor[HTML]{1f77b4}{left y-axis \textsc{FBD}} is computed between the generated sequences and the data distribution conditioned on the target class. For the first class,``PNG", functional sequences of \citet{Taskiran2023cell} are available, and we show their FBD.  The \textcolor[HTML]{ff7f0f}{right y-axis \textsc{Prob.}} refers to the target class probability of a classifier for the generated sequences in percent.}
    \label{fig:per-class-guidance}
    %\vspace{-\baselineskip}
\end{figure}
% As a reference, we also provide the FBD of sequences generated uniformly at random. 
\textbf{Q1: How well can Dirichlet FM capture the sequence distribution?}
We compare with an autoregressive language model (the best baseline in the promoter design experiments in Section \ref{sec:promoter-experiments}) and with Linear FM. To evaluate, we calculate the FBD between the models' generated sequences and the unconditional data distribution. Dirichlet FM outperforms the language model by a large margin on both datasets and linear FM for human melanoma cell enhancer generation  (Table \ref{tab:enhancer-design}). Moreover, distillation minimally impacts FBD while speeding up inference by \emph{3 orders} of magnitude compared to the language model and 2 to Dirichlet FM (distilled Dirichlet FM only requires 1 step). Such speedups, compared to autoregressive models, make Dirichlet FM a promising direction for other applications with high sequence lengths where inference times are important.

% \emph{Unconditional} generation via classifier-free guidance toward \emph{all} classes (further explained in Q3) improves Dirichlet FM's distribution match further with a large gap to the baselines on both datasets. 

\textbf{Q2: Can guided Dirichlet FM produce class-specific sequences and improve upon the state-of-the-art?}
We condition Dirichlet FM on different target cell-type classes via guidance (Section \ref{sec:method-guidance}). To quantify how well the generated sequences match the target class distribution, we use the FBD between the generated distribution and the data distribution \emph{conditioned} on the target class. Additionally, we train a separate cell-type classifier and evaluate the probability it assigns to the target class for a generated sequence. 

Sequences of classifier-free guided Dirichlet FM for the cell-type perineurial glia (PNG) have better FBD (Figure \ref{fig:per-class-guidance}) than the sequences of \citet{Taskiran2023cell}, of which several were \emph{experimentally validated} as functioning enhancers (we show this comparison only for the PNG class since their sequences are available for it). For the other classes, guidance is similarly effective; by increasing the guidance factor $\gamma$, the classifier probability for generated sequences to belong to the target class can be improved close to 100\%, and the FBD improves significantly until reaching a minimum (after the minimum the diversity decreases and the FBD worsens). The improvements with classifier guidance (Appendix Figure \ref{fig:per-class-guidance-cls}) are still significant but smaller.

\begin{figure}
    \centering
    \includegraphics[width=0.4\textwidth]{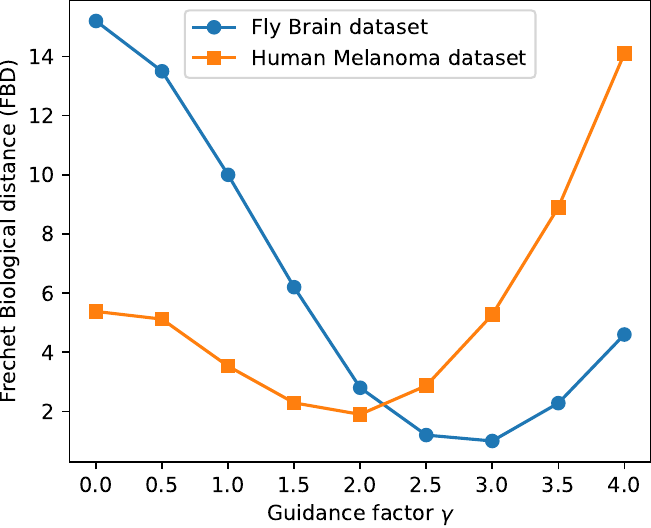}
    \vspace{-0.5\baselineskip}
    \caption{\textbf{Classifier-free guidance for \emph{unconditional} enhancer generation.} We generate unconditional sequences with classifier-free guidance by sampling the target class based on its empirical frequency. We show the FBD between the generated data and the full data distribution for varying levels of guidance $\gamma$. }
    \label{fig:guidance-weigth}
    %\vspace{-0.5\baselineskip}
\end{figure}

\textbf{Q3: Can guidance improve \emph{unconditional} generation?} We generate unconditional sequences with classifier-free guided Dirichlet FM by first sampling a class (based on its empirical frequency) and then guiding toward that class. This significantly improves sample quality compared to unguided Dirichlet FM (Figure \ref{fig:guidance-weigth}) and baselines (Table \ref{tab:enhancer-design}). Thus, guided Dirichlet FM via the connection we derive between flow and score reproduces the success in image diffusion models of enhancing sample quality via conditioning \citep{rombach2022high, saharia2022photorealistic}. Furthermore, like guidance for images \citep{xu2023restart}, the FBD first improves under increased guidance (Figure \ref{fig:guidance-weigth}) until reaching a minimum, after which increased guidance deteriorates sample diversity and, therefore, FBD.

\section{Conclusion}

We presented Dirichlet flow matching for modeling discrete data via a generative process on the simplex. This solves the pathological behavior of linear flow matching on the simplex, which we identified. Compared with autoregressive methods or diffusion with discrete noise, Dirichlet FM enables distillation and conditional generation via guidance, which we derived via a connection between the marginal flow and the scores of a mixture of Dirichlets.

Experimental results on important regulatory DNA sequence design tasks across 3 datasets demonstrate Dirichlet FM's effectiveness and utility for hard generative modeling tasks over long sequences. The results confirm Dirichlet FM's superiority to linear FM and multiple discrete diffusion approaches. Distilling Dirichlet FM only marginally impacts performance while enabling one-step generation, leading to multiple orders of magnitude speedups compared to autoregressive generative models for long sequences. Lastly, we demonstrated effective class conditional generation via guided Dirichlet FM to design cell-type specific enhancers - an important task for gene therapies. Hence, Dirichlet FM is a flexible framework (guidance, distillation) with strong performance for biological sequence generation and a promising direction for further discrete data applications.

%We anticipate that Dirichlet FM will be useful for further applications beyond DNA sequence design.
%\clearpage
\section*{Acknowledgements}
We thank Andrew Campbell, Yaron Lipman, Felix Faltings, Jason Yim, Ruochi Zhang, Rachel Wu, Jason Buenrostro, Bernardo Almeida, Gokcen Eraslan, Ibrahim I. Taskiran, and Pavel Avdeyev for helpful discussions.

This work was supported by the NSF Expeditions grant (award 1918839: Collaborative Research: Understanding the World Through Code), the Machine Learning for Pharmaceutical Discovery and Synthesis (MLPDS) consortium, the Abdul Latif Jameel Clinic for Machine Learning in Health, the DTRA Discovery of Medical Countermeasures Against New and Emerging (DOMANE) Threats program, the DARPA Accelerated Molecular Discovery program, the NSF AI Institute CCF-2112665, the NSF Award 2134795, the GIST-MIT Research Collaboration grant, the National Institute of General Medical Sciences of the National Institutes of Health under award number 1R35GM141861-01 and the U.S. Department of Energy, Office of Science, Office of Advanced Scientific Computing Research, Department of Energy Computational Science Graduate Fellowship under Award Number DE-SC0022158.

\section*{Impact Statement}
We present Dirichlet flow matching, an iterative generative modeling process for discrete data, and employ it for controllable sequence generation. For this general framework, there are many possible societal impacts, none of which need specific highlighting. We apply our method to regulatory DNA design. Current societal impacts of this application are overwhelmingly positive, with use cases in discovering biological regulatory mechanisms, drug development, gene therapy, and other therapies.

\bibliography{references}
\bibliographystyle{icml2024}

%%%%%%%%%%%%%%%%%%%%%%%%%%%%%%%%%%%%%%%%%%%%%%%%%%%%%%%%%%%%%%%%%%%%%%%%%%%%%%%
%%%%%%%%%%%%%%%%%%%%%%%%%%%%%%%%%%%%%%%%%%%%%%%%%%%%%%%%%%%%%%%%%%%%%%%%%%%%%%%
% APPENDIX
%%%%%%%%%%%%%%%%%%%%%%%%%%%%%%%%%%%%%%%%%%%%%%%%%%%%%%%%%%%%%%%%%%%%%%%%%%%%%%%
%%%%%%%%%%%%%%%%%%%%%%%%%%%%%%%%%%%%%%%%%%%%%%%%%%%%%%%%%%%%%%%%%%%%%%%%%%%%%%%
\newpage
\appendix
\onecolumn

\section{Method Details} \label{appx:method-details}

\paragraph{Flow Matching with Cross-Entropy Loss} For all $t$ and $\bfx \in S_K$, at convergence our denoising classifier satisfies
\begin{equation}
    \hat p(\bfx_1 \mid \bfx) = \frac{p_t(\bfx \mid \bfx_1) p_\text{data}(\bfx_1)}{p_t(\bfx)}
\end{equation}
Thus, if we parameterize the vector field via Equation \ref{eq:marginal-vectorfield}, then we are assured that
\begin{align}
    \hat v(\bfx, t;\theta) =  \sum_{i=1}^K u_t(\bfx \mid \bfx_1 = \bfe_i) \frac{p_t(\bfx \mid \bfx_1) p_\text{data}(\bfx_1)}{p_t(\bfx)} \quad = v(\bfx, t;\theta)
\end{align}

\paragraph{Proof of Proposition 1}
\begin{prop}
    Suppose that a flow matching model is trained with the linear flow map (Equation \ref{eq:linear-flow-map}). Then, for all $k = 2, \ldots K$ and $\bfx \sim p_t(\bfx)$, the converged model posterior $p(\bfx_1 \mid \bfx) \propto p_t(\bfx \mid \bfx_1)p_\text{data}(\bfx_1)$ has support over at most $k-1$ vertices for times $t > 1/k$.
\end{prop}
\begin{proof}
    In linear flow matching, the explicit form of the conditional probability path is given by \begin{equation}\label{eq:linear-prob-path}
        p_t(\bfx \mid \bfx_1 = \bfe_i) = \begin{cases} [(1-t)^{K-1}\Gamma(K)]^{-1} & x_i \ge t \\ 0 & x_i < t\end{cases}
    \end{equation}
    Suppose for sake of contradiction that $t > 1/k$ but $p_t(\bfx \mid \bfx_1)p_\text{data}(\bfx_1) \neq 0$ for $k' \ge k$ values of $\bfx_1$. Without loss of generality, suppose that $\bfe_1, \ldots \bfe_k$ are $k$ of those values. Then, by Equation \ref{eq:linear-prob-path}, we have $x_i \ge t$ for $i = 1, \ldots k$. Then
    \begin{align}
        \boldsymbol{1}^T\bfx = \sum_{i=1}^K x_i \ge \sum_{i=1}^k x_i \ge kt > k \cdot \frac{1}{k} = 1
    \end{align}
    This contradicts the fact that $\bfx$ must lie on the simplex.
\end{proof}
\allowdisplaybreaks 
\subsection{Dirichlet Conditional Vector Field}\label{appx:vf-derivation}

As preliminaries, we recall the definition of the multivariate beta function:
\begin{equation}\label{eq:beta-function}
    \mathcal{B}(\alpha_1, \ldots \alpha_K) = \frac{\prod_{i=1}^K \Gamma(\alpha_i)}{\Gamma\left(\sum_{i=1}^K \alpha_i\right)}  
\end{equation}
and that of the incomplete (two-argument) beta function:
\begin{equation}
    \mathcal{B}(x; a, b) = \int_0^x t^{a-1}(1-t)^{b-1} \, dt
\end{equation}
with the identity $\mathcal{B}(a,b) = \mathcal{B}(1; a, b)$.

We wish to construct a conditional flow $u_t(\bfx \mid \bfx_1 = \bfe_i)$ which generates the evolution of the conditional probability path
\begin{equation}
    p_t(\bfx \mid \bfx_1 = \bfe_i) = \Dir(\bfx; 1, \ldots \alpha_i = 1 + t, \ldots 1) = \frac{\Gamma(t+K)}{\Gamma(t+1)} x_i^{t}
\end{equation}
We choose the following \emph{ansatz} for the functional form of $u_t$:
\begin{equation}
    u_t(\bfx \mid \bfx_1 = \bfe_i) = C(x_i, t) (\bfe_i -\bfx)
\end{equation}
i.e., (1) the flow points towards the target vertex $\bfe_i$ and (2) the magnitude is scaled by a constant dependent only on $x_i$ and $t$. Then consider the $(K-2)$-dimensional hyperplane $A_b$ of constant $x_i = b$ which cuts through the simplex, separating it into two regions $S_1, S_2$ (Figure~\ref{fig:derivation}). We make the following key observation:

% In this work we use the convention that the density of the Dirichlet distribution has an additional factor of $1/\sqrt{K}$ compared to its usual definition:
% \begin{equation}
%     \Dir(\bfx, \boldsymbol{\alpha}) = \frac{1}{\sqrt{K}} \frac{1}{\mathcal{B}(\alpha_1, \ldots \alpha_K)} \prod_{i=1}^K {x_i}^{\alpha_i - 1}
% \end{equation}
% This is because our densities our normalized by performing a surface integral over $\mathcal{S}_K$, which is a $(K-1)$-dimensional regular simplex of edge length $\sqrt{2}$ rather 

\begin{center}
\emph{The probability flux crossing the plane $A_b$ is equal to the the rate of change of the total probability of region $S_1$.}
\end{center}
Thus, we solve for the constant $C(x_i, t)$ in the \emph{ansatz} by deriving these two quantities and setting them equal to each other.

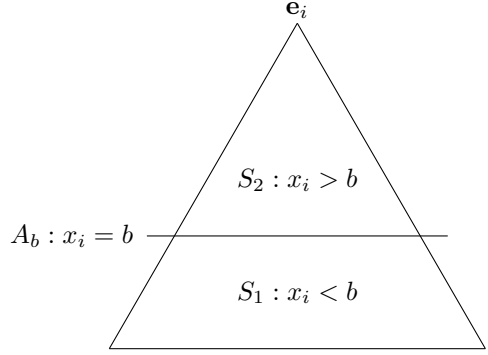
\begin{figure}
    \centering
    \begin{tikzpicture}
    \draw (-2.5, 0) -- (2.5, 0);
    \draw (-2.5, 0) -- (0, 4.33);
    \draw (2.5, 0) -- (0, 4.33);
    \node at (0, 4.5) {$\bfe_i$};
    \draw (-2, 1.5) -- (2, 1.5);
    \node at (-3, 1.5) {$A_b: x_i = b$};
    \node at (0, 0.75) {$S_1: x_i < b$};
    \node at (0, 2.25) {$S_2: x_i > b$};
    % \node at (-3.25, 2.25) {$\iint u_tp_t \cdot dA = -\del{}{t} p(S_1)$};
    \end{tikzpicture}
    \caption{Conceptual derivation of the conditional vector field} 
    \label{fig:derivation}
\end{figure}

\textbf{Q1: What is the probability mass of $S_1$ and its rate of change?}

The probability mass of $S_1$ can be obtained by integrating the density over each hyperplane $A_c$ (defined by $x_i=c$) for $c < b$ and then integrating over all such hyperplanes. Since the density is a constant proportional to $x_i^t = c^t$ over each hyperplane, we obtain
\begin{equation}
    p_t(S_1) \propto \int_0^b c^t \cdot \Vol(A_c) \, dc 
\end{equation}
% \begin{equation}
%     A_c = \left\{\bfx = (x_j)_{j\neq i} \in \mathbb{R}^{K-1} \bigg\vert \sum_{j \neq i}^K x_j = 1-c, x_j \ge 0 \right\}
% \end{equation}

where $\Vol(A_c)$ refers to the volume of the intersection between the hyperplane and the simplex. Since this region is defined by $x_i = c$, the remaining entries of $\bfx$ must add up to $1-c$. Thus, $A_c$ can be viewed as a nonstandard probability simplex over $K-1$ variables. The volume of this $(K-2)$-dimensional space is proportional to $(1-c)^{K-2}$, giving
\begin{equation}
    \rho(S_1) \propto \int_0^b c^t (1-c)^{K-2} \, dc = \mathcal{B}(b; t+1, K-1)
\end{equation}
Normalizing, we obtain
\begin{equation}
    p_t(S_1) = \frac{\int_0^b c^t(1-c)^{K-2} \, dc}{\int_0^1 c^t(1-c)^{K-2} \, dc} = \frac{\mathcal{B}(b; t+1, K-1)}{\mathcal{B}(t+1, K-1)} = I_b(t+1, K-1)
\end{equation}
where $I$ is the so-called \emph{regularized incomplete beta function} and is well-known as the CDF of the Beta distribution. Its derivative with respect to the first parameter is not available in closed form, but we write it as
\begin{equation}\label{eq:def_itilde}
    \tilde I_x(a, b) = \del{}{a}I_x(a, b) = \del{}{a}\frac{\mathcal{B}(x; a, b)}{\mathcal{B}(a, b)}
\end{equation}
and thus obtain the rate of change of $p_t(S_1)$ as $\tilde I_b(t+1, K-1)$.

\textbf{Q2: What is the probability flux across the hyperplane $A_b$?} The probability flux across $A_b$ (into $S_2$) is given by
\begin{equation}
    J = \iint_{A_b} u_tp_t \cdot \mathbf{\hat{n}} \, dA = \iint_{A_b} p_t \cdot C(b, t)  (\bfe_i - \bfx) \cdot \frac{\mathbf{n}}{\lVert \mathbf{n}\rVert} \, dA
\end{equation}
The normal vector $\mathbf{n}$ points from the center of the face opposite $\bfe_i$, specified by $x_i = 0$ and $x_j = 1/(K-1), j\neq i$, towards $\bfe_i$. Thus, the probability flux density is given by the dot product of
\begin{equation}
\begin{aligned}
    \bfe_i - \bfx &= (-x_1, \ldots \bfx_i = 1-b, \ldots -x_K) \\
    \mathbf{n} &= \left(-\frac{1}{K-1}, \ldots \bfx_i = 1, \ldots -\frac{1}{K-1}\right)
\end{aligned}
\end{equation}
Since $\sum_{j\neq i} x_j = 1-b$, this dot product is equal to $(1-b)K/(K-1)$. We also see that $\lVert \mathbf{n} \rVert = \sqrt{K/(K-1)}$. Importantly, the flux density is constant on the hyperplane $A$, so the total flux is given by a simple product:
\begin{align}
    J &= C(b, t) \cdot p_t \cdot  (1-b)\sqrt{\frac{K}{K-1}} \Vol(A_b) \\
    &= C(b, t) \cdot \left[\frac{1}{\sqrt{K}} \frac{\Gamma(t+K)}{\Gamma(t+1)} b^{t} \right] (1-b)\sqrt{\frac{K}{K-1}} \frac{\sqrt{K-1}}{\Gamma(K-1)}(1-b)^{K-2} \\
    &= C(b, t) \cdot \frac{(1-b)^{K-1}b^t}{\mathcal{B}(t+1, K-1)}
\end{align}
where (for this step only) we note that the volume of a simplex over $K$ variables \emph{viewed as a subset} of $\mathbb{R}^K$ is $\sqrt{K}/\Gamma(K)$, giving the Dirichlet PDF an additional factor of $1/\sqrt{K}$. 
Now setting $J = -\partial_t p_t(S_1)$ from above, we obtain
\begin{align}
    C(b, t) = -\tilde I_b(t+1, K-1) \frac{\mathcal{B}(t+1, K-1)}{(1-b)^{K-1}b^t}
\end{align}

\textbf{Checking the transport equation.} 
We check that $\nabla \cdot(p_t u_t) = -\partial p_t / \partial t$:% the divergence of this vector field under $\rho$:
\begin{align}
    \nabla \cdot (p_tu_t) &= - \nabla \cdot \left[ \frac{\Gamma(t+K)}{\Gamma(t+1)} x_i^{t}  \cdot \tilde I_{x_i}(t+1, K-1)\frac{\mathcal{B}(t+1, K-1)}{(1-x_i)^{K-1} x_i^t}(\bfe_i - \bfx)\right] \\
    &= - \Gamma(K-1) \nabla \cdot \left[ \frac{ \tilde I_{x_i}(t+1, K-1)}{(1-x_i)^{K-1} }(\bfe_i - \bfx)\right] \\
    &= - \Gamma(K-1) \left[  \frac{ \tilde I_{x_i}(t+1, K-1)}{(1-x_i)^{K-1} } \nabla \cdot (\bfe_i - \bfx) + (\bfe_i - \bfx) \cdot \nabla \frac{ \tilde I_{x_i}(t+1, K-1)}{(1-x_i)^{K-1} }  \right]\\
    \intertext{The divergence of $\bfx$ on the $(K-1)$-dimensional space $S_K$ is $K-1$. Also, the gradient of a function dependently only on $x_i$ has nonzero component only in the $x_i$ direction.}
    &= - \Gamma(K-1) \left[\frac{ \tilde I_{x_i}(t+1, K-1)}{(1-x_i)^{K-1} }(1-K)  + (1-x_i) \del{}{x_i}\frac{ \tilde I_{x_i}(t+1, K-1)}{(1-x_i)^{K-1} }  \right]\\
    &= - \Gamma(K-1) \left[\frac{ \tilde I_{x_i}(t+1, K-1)}{(1-x_i)^{K-1} }(1-K)  + (1-x_i)\left( \tilde I_{x_i} \frac{K-1}{(1-x_i)^K} + \frac{1}{(1-x_i)^{K-1}}\del{\tilde I_{x_i}}{x_i}  \right)\right]\\
    \intertext{The first and second terms now cancel.}
    &= - \Gamma(K-1) \left[ \frac{1}{(1-x_i)^{K-2}}\del{}{x_i}\tilde I_{x_i}(t+1, K-1)  \right]\\
    \intertext{Substituting Equation~\ref{eq:def_itilde} and interchanging the order of derivatives,}
    &= - \Gamma(K-1) \left[ \frac{1}{(1-x_i)^{K-2}}\del{}{t} \frac{1}{\mathcal{B}(t+1, K-1)}\del{}{x_i}\mathcal{B}(x_i; t+1, K-1) \right]\\
    &= - \frac{1}{(1-x_i)^{K-2}}\del{}{t} \frac{ \Gamma(K-1) }{\mathcal{B}(t+1, K-1)}x_i^t (1-x_i)^{K-2} \\
    &= -\del{}{t} \frac{ \Gamma(t+K) }{\Gamma(t+1)}x_i^t \quad = -\del{p_t}{t}
\end{align}

\section{Experimental details}

\subsection{Training and Inference}\label{appx:train-inference}

\begin{algorithm}[H]
\caption{\textsc{Training.}}\label{alg:training}
\begin{algorithmic}
\STATE \textbf{Input:} $N$ training sequences of one-hot vectors $\{\bfx^{(n)}\}$ for $n \in \{1\dots N\}$, neural network $\phi$ that takes a sequence as input and predicts probabilities for each sequence position.
\FORALL{$\bfx^{(n)}$}
\STATE Sample $t \sim Exp(1)$ 
\FORALL{sequence position $l \gets 0$ to $L - 1$}
\STATE Sample noisy $\tilde \bfx^{(n)}_l \sim p_t(\tilde \bfx^{(n)}_l \mid \bfx = \bfx^{(n)}_l)$ where $p_t(\tilde \bfx^{(n)}_l \mid \bfx = \bfx^{(n)}_l) = \Dir(\tilde \bfx^{(n)}_l; \boldsymbol{\alpha} = \boldsymbol{1} + t\cdot\bfx^{(n)}_l)$ 
\ENDFOR
\STATE Predict $\hat \bfp \gets \phi(\tilde \bfx^{(n)}, t)$ 
\STATE  Optimize loss $\mathcal{L} = \frac{1}{L}\sum_{l=0}^L\CrossEntropy(\hat \bfp_l, \bfx^{(n)}_l)$ 
\ENDFOR
\end{algorithmic}
\end{algorithm}

\begin{algorithm}[H]
\caption{\textsc{Inference.}}\label{alg:inference}
\begin{algorithmic}
\STATE \textbf{Input:} Trained neural network $\phi$. Number of integration steps $N$. Sequence length $L$. Maximum integration time $t_{max}$.
\STATE Sample $\bfx_0$ to be a sequence of $L$ probability vectors with each probability vector $\bfx_{0,l}$ sampled as $ \bfx_{0,l}\sim \Dir(\bfx; \boldsymbol{1}$)
\STATE Set $\bfx_t \gets \bfx_0$
\FOR{$n \gets 0$ to $N - 1$}
\STATE    Let $t \gets t_{\max} \cdot n / N$\;
\STATE    Predict $\hat \bfp \gets \phi(\bfx_t, t)$ \;
\FOR{sequence position $l \gets 0$ to $L - 1$}
\STATE    With $\bfx_{t,l,k}$ being the $k$-th entry of the probability vector $\bfx_{t,l}$: 
\FOR{simplex dimension $k \gets 0$ to $K - 1$}
\STATE Calculate conditional vector field $u_t(\bfx_{t,l} \mid \bfx_{1} = \bfe_k) = -\tilde I_{\bfx_{t,l,k}}(t+1, K-1)\frac{\mathcal{B}(t+1, K-1)}{(1-\bfx_{t,l,k})^{K-1} \bfx_{t,l,k}}(\bfe_k - \bfx_{t,l})$\;
\ENDFOR
\STATE     Calculate marginal vector field $u_t,(\bfx_{t,l}) = \sum_{k=0}^K \hat \bfp_{l,k} u_t(\bfx_{t,l} \mid \bfx_{1} = \bfe_k)$\;
\STATE    Take step $\bfx_{t,l} \gets \bfx_{t,l} + \frac{t_{\max}}{N} \cdot u_t(\bfx_{t,l})$ \;
\ENDFOR
\ENDFOR
\STATE \textbf{return} $\bfx_t$
\end{algorithmic}
\end{algorithm}

We provide algorithms for training and inference in Algorithm \ref{alg:training}, and Algorithm \ref{alg:inference}. These describe the procedures for generating discrete data on a single simplex. When generating a discrete data sequence, they are extended easily by performing the same procedure for a sequence of simplices. Code is at \url{https://github.com/HannesStark/dirichlet-flow-matching}.

\textbf{Toy experiments.} We train all models in Figure \ref{fig:toy-experiment} for 450,000 steps with a batch size of 512 to ensure that they have all converged and then evaluate the KL of the final step.

\textbf{Promoter Design.} We follow the setup of \citet{avdeyev2023dirichlet} and train for 200 epochs with a learning rate of $5\times 10^{-4}$ and early stopping on the MSE on the validation set. We communicated with \citet{avdeyev2023dirichlet} to ensure that we have the same training and inference setup as them, and we build on their codebase to evaluate the generated sequences with the Sei regulatory activity prediction model \citep{chen2022sei}. Thus, we use 100 inference steps for our Dirichlet FM instead of the 400 that they use in their code since they state that they used 100 integration steps for the results in the paper, which is also stated in the paper. Under this setup, we obtain the performance numbers for Linear FM and Dirichlet FM. We also ran the autoregressive language model in this setup (except for the number of inference steps which does not apply). Meanwhile, the numbers we report for Bit Diffusion, D3PM, and DDSM are taken from the DDSM paper \citep{avdeyev2023dirichlet}.

\textbf{Enhancer Design.} For both evaluations, on the human melanoma cell and the fly brain cell dataset, we train for 800 epochs (convergence of validation curves is reached after approximately 300 for both datasets). We use FBD for early stopping. For inference, we use 100 integration steps.

\textbf{Classifier-free Guidance.} For classifier-free guidance, we train with a conditioning ratio (the fraction of times we train with a class label as input instead of the no-class token as input) of 0.7. During inference, the inference Algorithm \ref{alg:inference} is changed in that two probability vectors $\hat \bfp$ are predicted, once with class conditioning and once without class conditioning, which we sum together according to Equation \ref{eq:cfg}. Then, we project the resulting probabilities onto the simplex since negative values can arise for guidance factors $\gamma>1$. For this purpose, we use the algorithm by \citet{wang2013projection}. 

\textbf{Classifier Guidance.} The noisy classifier that we train for classifier guidance has the same architecture as our generative model, except that we sum the final representations and feed them into a 2-layer feed-forward network that serves as classification head. For training, we use early stopping on the accuracy and train for approximately 800 epochs. During inference, we use automatic differentiation to obtain the classifier's gradients with respect to the input points on the simplices. To perform classifier guidance, we then convert the flow model output probabilities to scores as described in Section \ref{sec:method-guidance}, obtain the guided score by adding the unconditional and the conditional scores via Equation \ref{eq:classifier-guidance}, and convert the obtained scores back to probabilities. These we project to the simplex \citep{wang2013projection} from which we obtain the vector field for integration. 

\textbf{Classifier for FBD calculation.} This classifier's architecture is similar to that of the noisy classifier for classifier guidance. However, it does not have any time conditioning and takes token embeddings as input instead of points on the simplex. For training, we use early stopping on the accuracy and train approximately 100 epochs.

The sequence embeddings that we use to calculate FBD are given by the hidden features after the first layer of the classification head. The 4 classes that we choose for the cell type specific enhancer generation are chosen as classes with a good tradeoff between the area under the receiver operator characteristic curve and the area under the precision-recall curve, and er chose the class PNG for which \citet{Taskiran2023cell} had data available. 

\textbf{Distillation.} For distillation, we run inference with the teacher model for every training step of the student model to obtain pairs of noise and training targets. With this, we train the student model on approximately 6 billion sequences. 

\textbf{Computational requirements.} We train on RTX A600 GPUs. Training in the enhancer generation setup for 200 epochs on sequences with length 500 takes 7 hours. Our largest model uses approximately 8GB of RAM during training.

\textbf{Architecture}
The architecture that we use for the promoter design experiments is the same as in DDSM \citep{avdeyev2023dirichlet}. In our other experiments we replaced group norm with layer norm in their architecture. The model consists of 20 layers of 1D convolutions interleaved with time embedding layers (and class-type embedding layers for classifier-free guidance) and normalization layers. We also experimented with Transformer architectures, which led to worse performance. All models use this 20-layer architecture except for the classifier for the Fly Brain data, which has 5 layers.

\subsection{Data} \label{appx:data}
\begin{figure}
    \centering
    \includegraphics[width=0.45\textwidth]{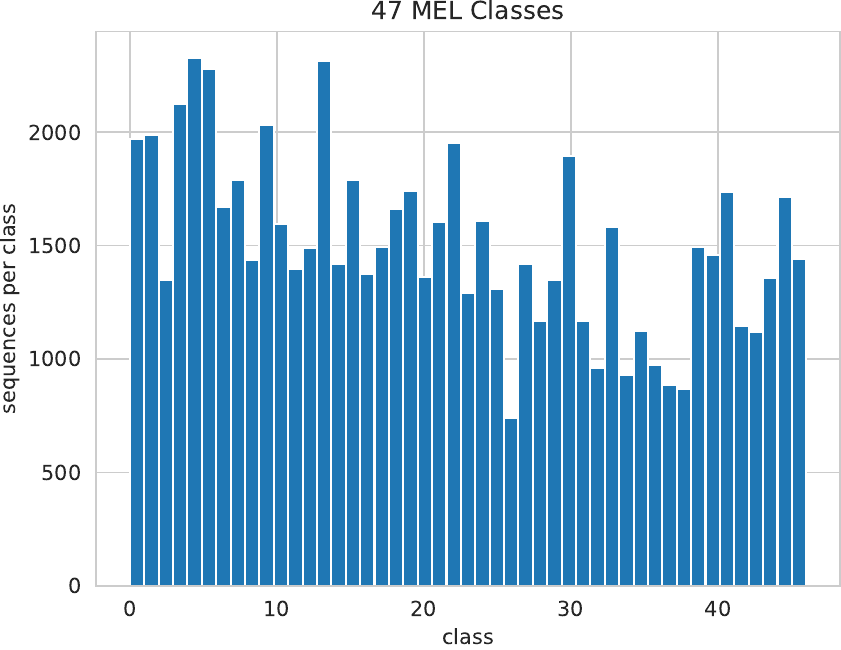}
    \includegraphics[width=0.45\textwidth]{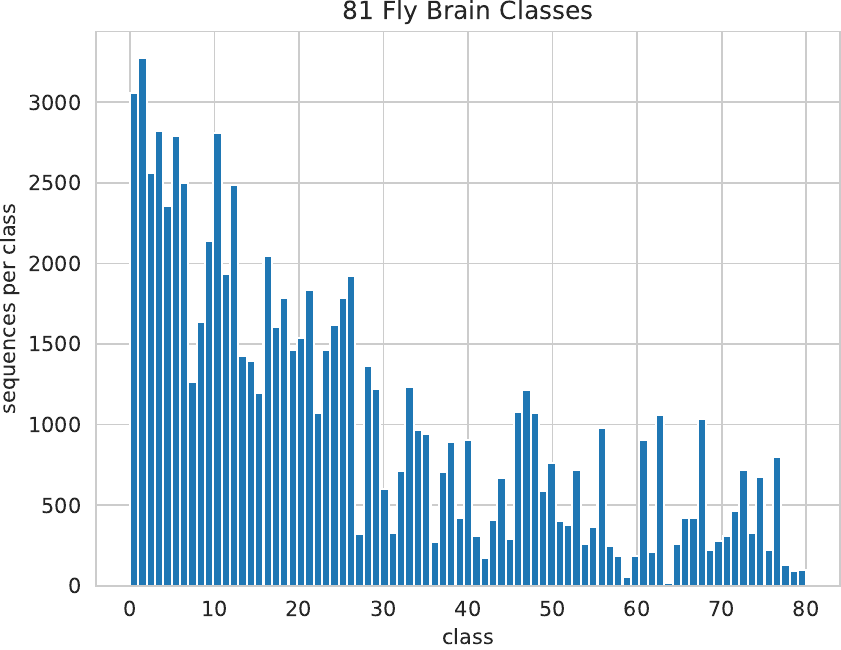}
    \caption{{Histograms of cell type classes for fly brain cell enhancer sequence data and human melanoma cell data.}}
    \label{fig:class-hist}
\end{figure}
For the enhancer data of 104665 fly brain cell sequences \citep{janssens2022deepflybrain}, we use the same split as \citet{Taskiran2023cell}, resulting in an 83726/10505/10434 split for train/val/test. Meanwhile, for the human melanoma cell dataset of 88870 sequences \citep{atak2021deepmel2}, their split has 70892/8966/9012 sequences. It is noteworthy that these datasets contain ATAC-seq data \citep{buenrostro2013atacseq}, which means that not all sequences are guaranteed to be enhancers and actually enhance transcription of a certain gene. ATAC-seq only measures the chromatin accessibility of the sequences in the cell types, which is a necessary but not sufficient requirement for a sequence to be an enhancer. In Figure \ref{fig:class-hist}, we show histograms for the class distributions of both datasets. For the melanoma dataset, there is little class imbalance.

%\section{Additional Related Work}
%minkais flow matching
%Retro bridge
%Latent diffusion for DNA sequences does not have code

\section{Additional Results}

\begin{figure}
    \centering
    \includegraphics[width=0.49\textwidth]{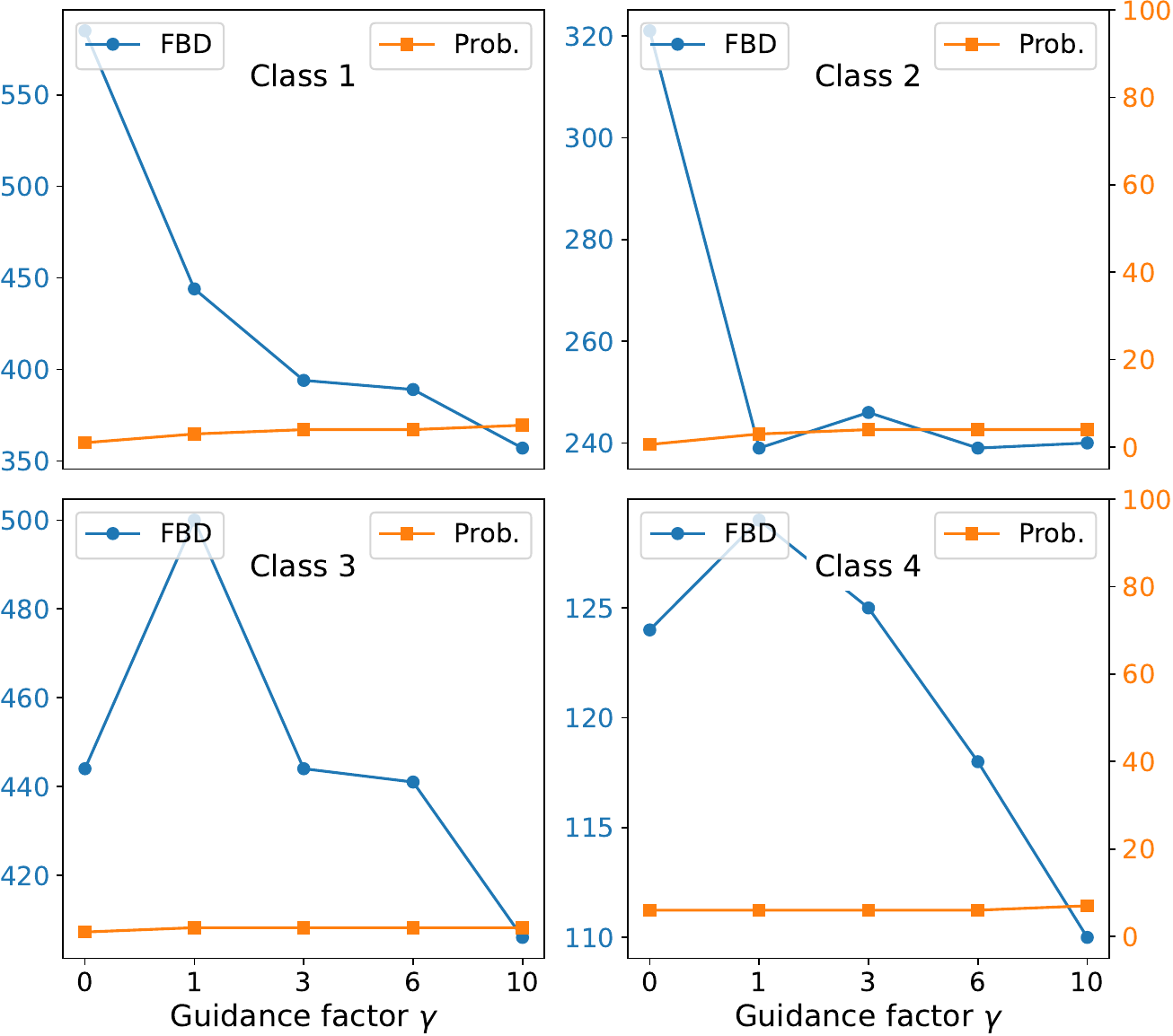}
    \caption{\textbf{Evaluation for cell type specific enhancer design with classifier guidance.} We condition Dirichlet FM to generate sequences that are only active in a class of cell type via classifier guidance with varying $\gamma$. Results are shown for 4 classes in the Fly Brain cell data. The left y-axis \textsc{FBD} is computed between the generated sequences and the data distribution conditioned on the target class. The right y-axis \textsc{Prob.} refers to the target class probability of a classifier for the generated sequences.}
    \label{fig:per-class-guidance-cls}
\end{figure}

\begin{figure}
    \centering
    \includegraphics[width=0.49\textwidth]{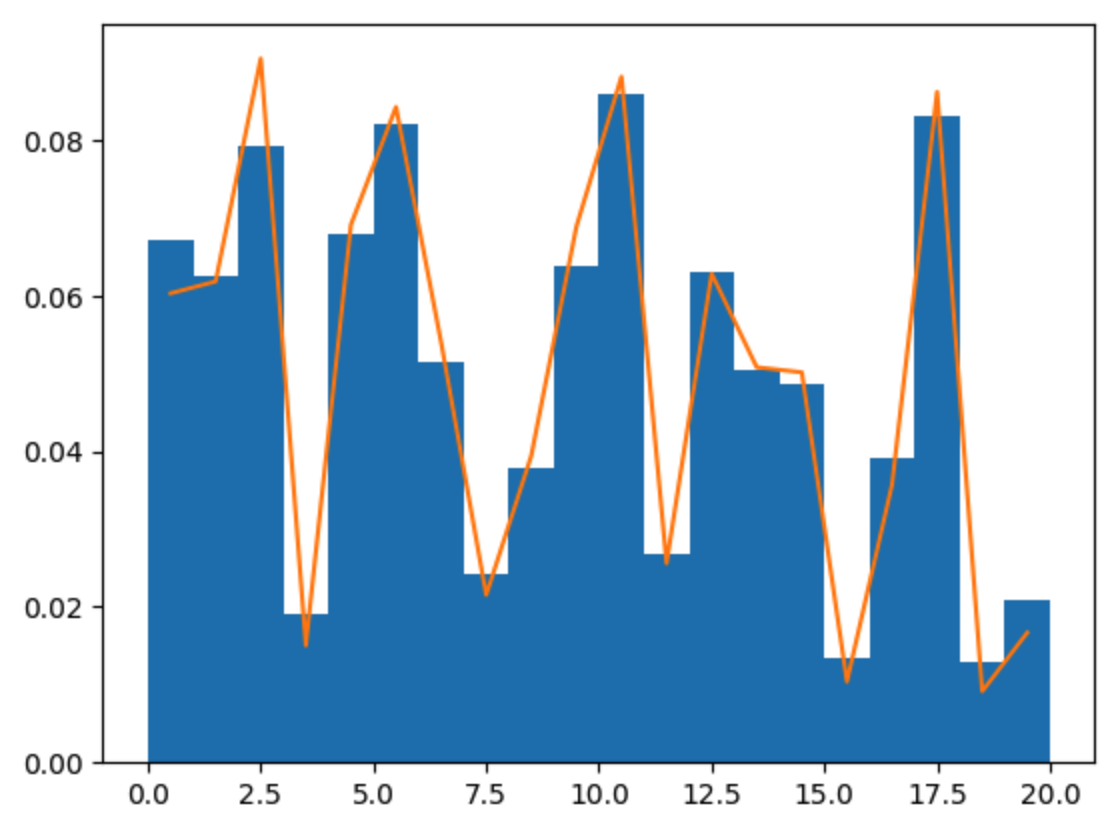}
    \caption{\textbf{Simulation results for classifier guidance on a conditional categorical distribution with an analytically tractable vector field.} In blue, we show a histogram (with frequencies normalized to density) of the generated classes for each class on the x-axis. In orange, we show the true probabilities of the toy distribution conditioned on the same class that we use for classifier guidance.}
    \label{fig:guidance-toy-experiment}
\end{figure}

\textbf{Analytical toy experiment for classifier guidance.}
As a toy experiment to demonstrate our classifier guidance procedure, we construct a distribution conditioned on a binary random variable. The conditional distribution is a categorical distribution over 20 classes. In this setup, the time-dependent class probabilities conditioned on a noisy point on the simplex can be computed analytically. Thus, we can use them to obtain the vector field for Dirichlet FM analytically. Furthermore, the class probabilities and the gradients of their logarithm (the score) can be computed analytically. Hence, we can simulate classifier guided Dirichlet FM analytically for this toy distribution.

The results in Figure \ref{fig:guidance-toy-experiment} show a close match between the empirical distribution of the generated data and the ground truth probabilities. This empirically confirms the effectiveness of our classifier guidance procedure that relies on converting probabilities to scores and converting them back to probabilities by solving a linear system of equations.

\begin{figure}
    \centering
    \includegraphics[width=0.45\textwidth]{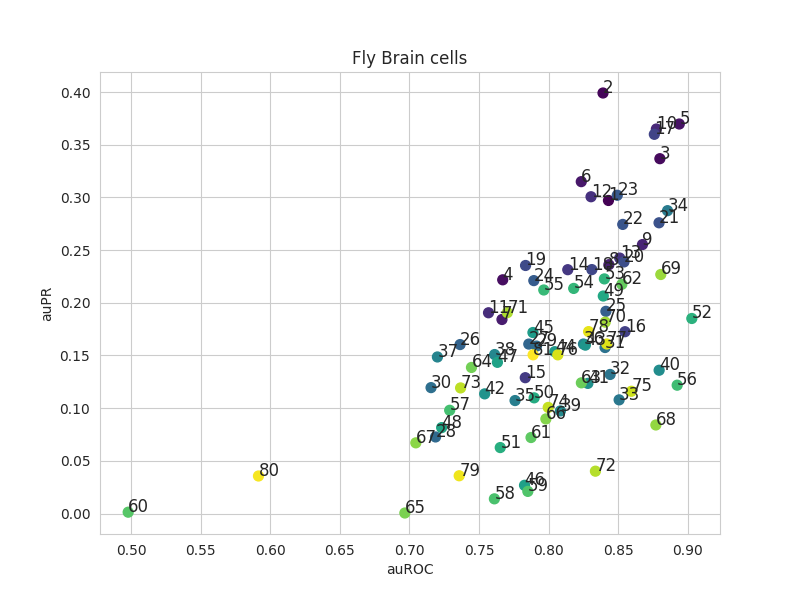}
    \includegraphics[width=0.45\textwidth]{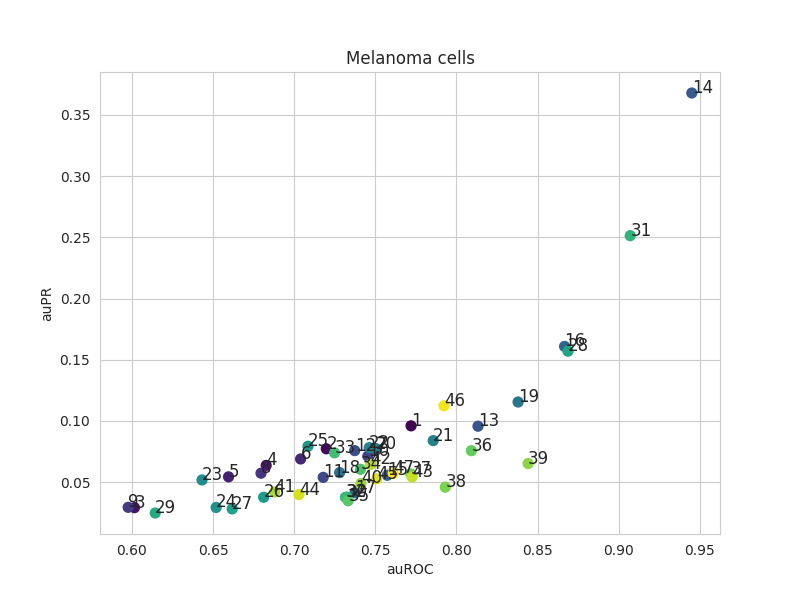}
    \caption{Per class performance of the classifiers similar to those used for evaluation and FBD calculation without early stopping. Shown are scatter plots for each class of the fly brain data (left) and the human melanoma cell data (right) between the area under the curve of the receiver operator characteristic (x-axis) and the area under the precision-recall curve (y-axis). The scatter plots match the results of the classifiers of \citet{atak2021deepmel2} and \citet{janssens2022deepflybrain} closely.}
    \label{fig:auroc_aupr}
\end{figure}

%%%%%%%%%%%%%%%%%%%%%%%%%%%%%%%%%%%%%%%%%%%%%%%%%%%%%%%%%%%%%%%%%%%%%%%%%%%%%%%
%%%%%%%%%%%%%%%%%%%%%%%%%%%%%%%%%%%%%%%%%%%%%%%%%%%%%%%%%%%%%%%%%%%%%%%%%%%%%%%

\end{document}